\documentclass[10pt,journal,compsoc]{IEEEtran}

\usepackage{amsfonts}
\usepackage[cmex10]{amsmath}
\allowdisplaybreaks[1]  
\usepackage{amsthm}
\usepackage{amssymb}
\usepackage{mdwmath}
\usepackage{mathtools}
\usepackage{cases}
\usepackage{empheq}

\theoremstyle{definition}

\newtheorem{theorem}{Theorem}

\usepackage{tabularx}
\usepackage{array}
\usepackage{mdwtab}
\usepackage{multirow}
\usepackage{booktabs}  
\usepackage{diagbox} 
\usepackage{makecell}

\usepackage{textcomp}
\def\BibTeX{{\rm B\kern-.05em{\sc i\kern-.025em b}\kern-.08em
T\kern-.1667em\lower.7ex\hbox{E}\kern-.125emX}}

\usepackage[pdftex]{graphicx}
\usepackage{caption}  
\usepackage[labelformat=simple]{subcaption}
\graphicspath{ {./img/} }

\usepackage{cite}

\usepackage{makeidx}
\usepackage{ragged2e}

\usepackage{algorithm}
\usepackage{algorithmic}

\usepackage{hyperref} 
\usepackage{cleveref} 

\newcolumntype{C}[1]{>{\centering\arraybackslash}m{#1}}

\hypersetup{colorlinks = true, 
	    linkcolor = black, 
	    urlcolor = black,
        citecolor = black} 

\begin{document}

\bstctlcite{IEEEexample:BSTcontrol}  

\title{\huge{AC$^2$P$^2$SL: Adaptive Communication-Computation Pipeline Parallel Split Learning over Edge Networks}}

\author{
\IEEEauthorblockN{
    Chenyu~Liu, 
    Zhaoyang~Zhang, 
    Zirui~Chen, 
    Zhaohui~Yang,
    Chunhui~Feng,
    and Tony Q. S. Quek
}

\IEEEcompsocitemizethanks{\IEEEcompsocthanksitem
A preliminary part of this work was presented at IEEE Globecom 2025 Workshop on A4E: AI/ML for Edge/Fog Networks \cite{C2P2SL}.
\IEEEcompsocthanksitem
This work was supported in part by National Natural Science Foundation of China under Grants 62394292 and 624B2129, Zhejiang Provincial Key R\&D Program under Grant 2023C01021, and the Fundamental Research Funds for the Central Universities under Grant 226-2024-00069. 
(\textit{Corresponding author: Zhaoyang Zhang}.)
\IEEEcompsocthanksitem
C.~Liu, Z.~Zhang, Z.~Chen, Z.~Yang, and C.~Feng are with the College of Information Science and Electronic Engineering, Zhejiang University, Hangzhou 310027, China, 
and also with the Zhejiang Provincial Laboratory of Multi-Modal Communication Networks and Intelligent Information Processing, Hangzhou 310027, China
(e-mail: chenyuliu@zju.edu.cn; ning\_ming@zju.edu.cn; ziruichen@zju.edu.cn; yang\_zhaohui@zju.edu.cn; chunhuifeng@zju.edu.cn).
\IEEEcompsocthanksitem
Tony Q. S. Quek is with Singapore University of Technology and Design,
Singapore 487372 (e-mail: tonyquek@sutd.edu.sg).
}
}


\IEEEtitleabstractindextext{%
\begin{abstract}
\justifying 
In wireless edge networks, split learning (SL) enables base station (BS) to utilize the distributed data and computing power across user equipments (UEs) to achieve collaborative model training while protecting local data privacy. However, the inherent sequential execution of computation and communication processes in conventional SL usually leads to long training times. To overcome this limitation, this paper proposes an adaptive communication-computation pipeline parallel split learning (AC$^2$P$^2$SL) framework. By conceptualizing the communication and computation processes of UEs and the BS as a unified pipeline, AC$^2$P$^2$SL achieves fine-grained pipeline parallelism across multiple micro-batches. Through this approach, effective overlapping of communication and computation is achieved which results in significant reduction of the overall training latency. Moreover, by considering the system constraints in the communication, computation, and storage dimensions as well as the heterogeneity of UEs, we formulate a joint optimization problem to minimize the training time and propose a corresponding split and pre-allocation algorithm to further enhance the pipeline efficiency. Additionally, accounting for the practical dynamic environments for the UEs, we design an adaptive re-allocation strategy to enhance the system resilience. Extensive experimental results demonstrate the effectiveness and robustness of AC$^2$P$^2$SL in reducing training time while ensuring data privacy preservation. 
\end{abstract}

}

\maketitle

\section{Introduction} \label{sec:introduction}

In the sixth-generation (6G) era, deep integration of communication networks and  artificial intelligence (AI) has become an inevitable trend \cite{wBAIM, DL}. In edge network scenarios, such as anomaly detection monitoring \cite{Monitor, Detection}, unmanned aerial vehicle (UAV) image capturing\cite{WDL, UAV, UAVGeo}, and vehicle autonomous driving\cite{VUSFL, VEC}, raw data for neural network training are often stored dispersedly on local user equipments (UEs). Currently, systems typically collaborate with multiple UEs to train a common model, thereby fully leveraging their individual data to enrich the overall training dataset. With the rapid development of deep learning (DL) technologies, the parameter scale and computational load of AI models are consistently increasing \cite{BAIM, law}. However, UEs in edge networks are typically resource-constrained personal devices, and relying solely on them to collaboratively train complex models presents multiple deficiencies in terms of computing power, storage, and communication capabilities\cite{device, Hivemind}. 

Meanwhile, wireless systems are also evolving, transitioning from a pure communication-oriented approach to an integrated communication and computing paradigm \cite{RAN}. For instance, AI-enabled radio access network (AI-RAN) architecture \cite{AI-RAN} has transformed traditional base stations (BS) into intelligent AI processing nodes, providing additionally computational resources for edge AI tasks. Consequently, a straightforward approach is to centralize data from UEs at BS for model training, yet this raises numerous privacy-related concerns \cite{DL, multihop}. By contrast, split learning (SL) \cite{SL, healthSL, DCSFL}, has emerged as an effective solution, where the model is split into two parts: the first few layers are distributed to the UEs, while the main body of the model is retained on the BS. For the split point, system utilizes wireless transmission to exchange forward propagation (FP) activations and backward propagation (BP) gradients between computing devices \cite{WSL}.
By splitting the model, the computational load is apportioned to the UEs and the BS, ensuring that computations directly involving user data are performed locally on the UEs without being exposed to the BS.

\begin{figure*}[t] 
\centering
\includegraphics[width=0.74\linewidth]{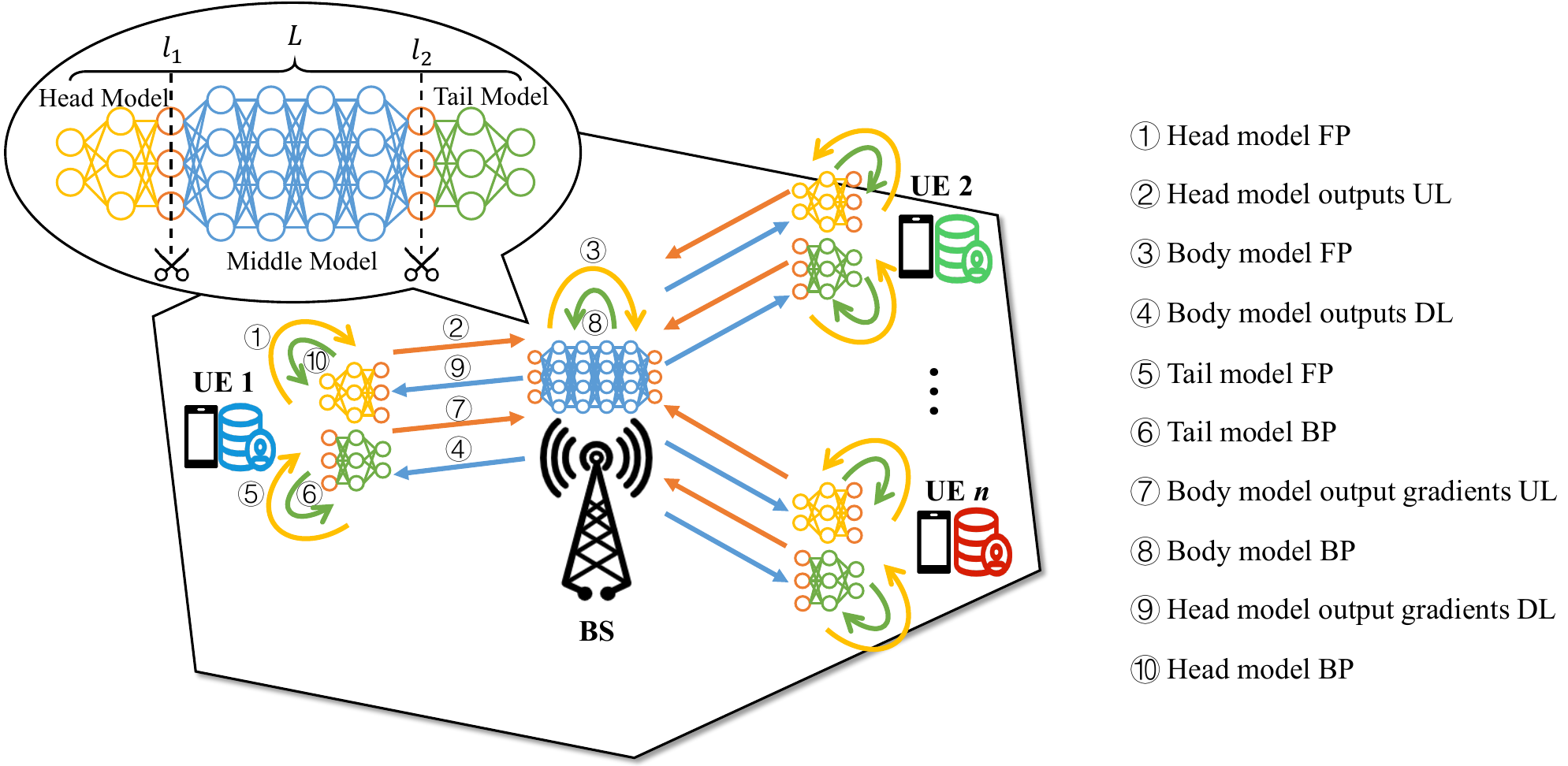}
\caption{U-shaped parallel split learning over wireless edge networks.}
\label{framework}
\end{figure*}

In many supervised learning tasks, such as image recognition \cite{USL} and wireless localization \cite{AL}, in addition to user data, the privacy of label information often also requires protection. To meet this demand, \cite{SL, USL, UPSL} introduces U-shaped SL (USL) characterized by a more ingenious two-layer cut strategy, which splits model into three sub-models: head, body, and tail models. UEs locally train head and tail models, while the BS trains the body model. 
This necessitates the BS serving solely as a computational unit for the intermediate model layers, without accessing either the input data or the final outputs/labels, thereby ensuring that both data, model outputs and ground-truth labels are strictly retained locally. However, while protecting privacy, SL and USL also introduce additional latency caused by wireless transmission\cite{splitMAC}, making the improvement of the overall temporal efficiency of the architecture an important research topic.

Furthermore, to facilitate simultaneous multi-user participation in SL, the split federated learning approach in \cite{SFL, ASFL, PipeSFL} extends the aforementioned framework by introducing the parallel model training mechanism among users. 
However, these methods typically require an additional aggregation server \cite{PipeSFL}, thereby introducing extra synchronization overhead and posing potential privacy risks. From another perspective, \cite{CPSL} partitions UEs into multiple clusters to implement intra-cluster parallel training and inter-cluster sequential training, partly reducing training latency. Furthermore, \cite{UPSL} introduced parallel split learning for multiple UEs, while \cite{SGLR, EPSL} reduces the dimensionality of BP by aggregating the gradients of the BS-side layer. These approaches eliminate the need for UE-side model synchronization, thus further decreasing communication time\cite{PSL}. 

{Moreover, in edge network SL, the selection of the model split layer significantly affects training efficiency. Additionally, intrinsic heterogeneity among terminal devices in computation and communication capabilities induce substantial synchronization latency during the training process. Building upon the proposed SL schemes, \cite{CPSL} designed a joint optimization strategy for layer split selection and resource allocation to minimize training costs. \cite{ASFL} and \cite{UPSL} jointly optimized the split layer selection and bandwidth allocation problems within the frameworks, respectively, to minimize training latency. Furthermore, \cite{EPSL} additionally introduced power control to balance the trade-off with energy consumption. However, in practical heterogeneous edge networks, not only do intrinsic performance disparities exist among devices, but the communication capabilities of UEs also fluctuate due to mobility and other factors, while their computational resources remain unstable due to concurrent local computing tasks. The synchronization latency induced by these varying factors significantly degrades overall training efficiency.}

For the communication overhead during the training process, above approaches have primarily focused on reducing data transmission time via operations such as UE-side parallelism and data aggregation. However, as shown in Figure \ref{framework}, for a single data batch, the computation and communication processes remain serially executed. From head model FP to head model BP, subsequent stages needs to wait for the completion of the previous stages, inevitably introducing idle time across the various processing stages. In scenarios involving models with high-dimensional intermediate outputs or poor channel conditions, the prohibitive communication overhead results in low training efficiency. To address these limitations, we notice that the serial situation can be significantly optimized through fine-grained parallelism, similar to pipeline parallelism mechanism \cite{GPipe} in distributed training. { By partitioning data batches into multiple micro-batches, each node propagates the results to the subsequent stage immediately after completing its sub-model computation on the current micro-batch, while concurrently processing intermediate outputs from the preceding stage for subsequent micro-batches. Stages in the computation pipeline can be parallelized through the overlap between micro-batches. Although recent works have introduced pipeline parallelism into SL, their scopes remain limited. Specifically, \cite{WPP} focuses exclusively on intra-server computational pipelining, while \cite{GAN} implements pipeline parallelism for the fixed stages within federated split learning framework. Critically, neither approach accounts for the inherent parallelizability between the communication and computation processes in wireless SL. This neglect of communication-computation parallelism restricts the overall training efficiency in dynamic wireless edge environments.}

In this paper, building upon USL, we integrate the communication pipeline consisting of uplink and downlink data transmission with the computation pipeline comprising sub-models' FP and BP into a unified training pipeline. Treating UEs' computation, uplink transmission, BS's computation, and downlink transmission as distinct processing stages, we also achieve fine-grained micro-batch parallelism by partitioning data batches. This approach enables wireless transmission to occur concurrently with local model computations on both UE and BS sides to achieve parallel training. 
Moreover, in response to heterogeneous capability, time-varying communication quality, and real-time availability of computational resources in practical edge networks, we formulate the pipeline training time optimization problem to design the split and pre-allocation (SPA) algorithm as well as the adaptive re-allocation (ARA) strategy. 
Consequently, the above content collectively constitutes an adaptive communication-computation pipeline parallel split learning (AC$^2$P$^2$SL) framework in this paper.
By effective parallelism between training stages, the proposed framework significantly reduces pipeline training time and enhances overall efficiency.
The main contributions of this paper are summarized as follows: 

\begin{itemize}
    \item We propose the AC$^2$P$^2$SL framework, which combines data transmission with computation tasks to achieve communication-computation pipeline parallelism between micro-batches by splitting data batches.
    \item We formulate a joint optimization problem to minimize the pipeline training time based on the computation, communication, and memory constraints of UEs and design the SPA algorithm for its solution.
    \item We implement an adaptive ARA strategy before each pipeline training round to adjust the micro-batch quantity, allocated batch size and time slot in response to significant changes in UE performance.
    \item We comprehensively evaluate the effectiveness and robustness of AC$^2$P$^2$SL through extensive experiments under different models and system parameters. We also conduct ablation studies to further validate the efficacy and principle of our resource allocation strategies.
\end{itemize}

The remainder of this paper is organized as follows. 
Section \ref{section2} introduces the system communication, computation, and storage model. 
In Section \ref{section3}, we present the overview of AC$^2$P$^2$SL and its training workflow. 
Then, problem formulation and solution approach are presented in Section \ref{section4}.
Next, Section \ref{section5} provides performance evaluations of our proposed scheme. 
Finally, we conclude our work in Section \ref{section6}.

\section{System Model}\label{section2}
We consider an edge cellular network consisting of a central BS and a set of $n$ UEs distributed within the coverage area. Table \ref{notations} summarizes the main notations of this paper.

\begin{itemize}
    \item \textbf{BS:} By Integrating AI processing units, the BS is equipped with sufficient computational resources to train large models. Beyond performing wireless signal transmission, the BS is responsible for collecting computational and communication information from UEs to make optimizations during the training process.
    
    \item \textbf{UE:} As resource-constrained edge devices, each UE has weak computational capability to train tiny neural networks. {Let $\mathcal{N}=\{1, \dots, N\}$ denote the set of all devices. The UEs locally store private datasets,} where $D_{i,j}$ represents the local data samples of the $i$-th UE during the $j$-th training batch.
\end{itemize}

\begin{table}[t]
  \centering
  \caption{Summary of Main Notations} \label{notations}
  \setlength{\extrarowheight}{0.3em}
  \begin{tabular}{lp{6.3cm}} 
    \toprule
    \textbf{Notation} & \textbf{Description} \\ 
    \midrule 
    $l_1$, $l_2$   &   Cutting layers of head, body, and tail sub-models\\
    $a_l$   &   Data size of the $l$-th layer’s activation output\\
    {$B$, $b_i$}   &   {Total and UE $i$'s batch size of one training round}\\
    $f_i$, $f_0$   &   Attainable computing performance of UE $i$ and BS\\
    $I_i$,$I_0$   &   Operational intensity of UE $i$ and BS\\
    $F_i$, $F_0$   &   Peak FLOPS of UE $i$ and BS\\
    $\beta_i$, $\beta_0$   &   Maximum memory bandwidth of UE $i$ and BS\\
    $c_l^f$, $c_l^b$   &   FP and BP computing workload per sample of the $l$-th layer\\
    $m_l^f$, $m_l^b$   &   Fixed memory access for FP and BP of the $l$-th layer\\
    $\Delta m_l^f$, $\Delta m_l^b$   &   Activation memory access for FP and BP per sample of layer $l$\\
    $u_l, \Delta u_l$   &   Fixed and per sample activation memory of the $l$-th layer\\
    $U_i$   &   Maximum memory limit of UE $i$\\
    $\tau$, $T$   &   Length of time slot and time frame\\
    $s_i$, $\rho$   &   Number of time slots for UE $i$ and ratio of uplink to downlink time slots\\
    {$BW$}   &   {System bandwidth}\\
    $G_i$, $G_0$   &   Effective antenna gain of UE i and BS\\
    $p_i$, $p_0$   &   Transmit power of UE $i$ and BS\\
    $h_i$   &   Channel gain between UE $i$ and BS \\
    $N_0$   &   Power spectral density of noise\\
    $k$   &   Number of micro-batches in a single data batch\\
    $x(i,j)$, $y(i,j)$   &   The $j$-th micro-batch data samples and labels of UE $i$\\
    ${a/g}_{h/t}(i,j)$   &   Head/tail model activation output/gradient of the $j$-th micro-batch of UE $i$\\
    $a_b(j)$, $g_b(j)$   &   Body model activation output and gradient of the $j$-th micro-batch of BS\\
    {$W_h(i)$, $W_t(i)$}   &   {Head, and tail model parameter of UE $i$}\\
    {$W_b$}   &   {Body model parameter of BS}\\
    {$\mathcal{F}(\cdot)$, $\mathcal{B}(\cdot)$}   &   {Forward and backward propagation mapping}\\
    {$\mathcal{L}(\cdot)$, $l_{i,j}$}   &   {Loss function and $j$-th micro-batch loss of UE $i$}\\
    \bottomrule
  \end{tabular}
\end{table}

\subsection{Computation Model}\label{section2.1}

The target model has a total of \(L\) layers, which are partitioned into head model, body model, and tail model before training. The indices of two cut layers are respectively denoted as $l_1, l_2 \in\{1,2,\ldots,L\}$. In this way BS retains the sub-model containing the majority of parameters, while broadcasting the smaller head and tail sub-models to individual UEs. During one training round, the batch size of $i$-th UE's input data is denoted as \(b_i\), and the total batch size of all UEs is given as
{\begin{equation}\label{batchbound}B=\sum_{i=1}^{N}b_i.\end{equation}}

In contrast to previous studies that ideally treat device computational capability as a static constant\cite{C2P2SL}, we use the Roofline model \cite{Roofline} to provide a more realistic characterization. This model formulates the attainable computing performance as a piecewise function determined by the hardware's peak floating point operations per second (FLOPS), memory bandwidth, and the data operational intensity. Here, operational intensity is defined as the ratio of the floating point operations (FLOPs) to the amount of memory access required during the computation process.

For layer $l$, the total memory access is composed of fixed parameter memory access $m_l^{f/b}$ and per sample memory access $\Delta m_l^{f/b}$, a variable component proportional to the number of input data samples. Depending on the context of FP or BP, the parameter $m_l^{f/b}$ and $\Delta m_l^{f/b}$ takes the value of $m_l^f,m_l^b$ or $\Delta m_l^f,\Delta m_l^b$, respectively. Consequently, for the varying lower layer bound $l_{L}$, upper layer bound $l_{U}$ of sub-model and the number of input data samples $D$, the operational intensity of the $i$-th device for FP or BP is expressed as 
\begin{equation}
    I^{f/b}(D,l_{L},l_{U}) = \frac{D \sum\limits_{l=l_{L}}^{l_{U}}c_l^{f/b}}{\sum\limits_{l=l_{L}}^{l_{U}} m_l^{f/b} + D \sum\limits_{l=l_{L}}^{l_{U}} \Delta m_l^{f/b}},
\end{equation}
where the superscript $f/b$ of operational intensity indicates whether the current computing task is FP or BP. $\sum_{l=l_{L}}^{l_{U}}c_l^{f/b}$ denote the FLOPs required for a single data sample. When inputting a single data sample into the neural network, the FLOPs of $l$-th layer for FP and BP are denoted by \(c_{l}^{f}\) and \(c_{l}^{b}\). 

Furthermore, we define the borderline operational intensity as $I_i^{\max} = \frac{F_i}{\beta_i}$, where $F_i$ represents the peak FLOPS and $\beta_i$ denotes the memory bandwidth. When the operational intensity falls below this threshold, the task is classified as memory-bound, which implies that the performance is constrained by the memory bandwidth, i.e. $f^{f/b}_i = \beta_i I^{f/b}(b_i,l_{L},l_{U})$. Conversely, when the intensity exceeds the threshold, the task is compute-bound, limited by the peak FLOPS, i.e. $f^{f/b}_i = F_i$. Therefore, the attainable computational capability of the $i$-th device for FP or BP can be summarized as

\begin{equation}\label{roofline}
    \!f^{f/b}_i(b_i,l_{L},l_{U}) = \min\! \left(\!F_i, \frac{\beta_i b_i \sum\limits_{l=l_{L}}^{l_{U}}c^{f/b}_l}{\sum\limits_{l=l_{L}}^{l_{U}} \! m^{f/b}_l + b_i \sum\limits_{l=l_{L}}^{l_{U}} \! \Delta m^{f/b}_l}\! \right)\!.
\end{equation}

\begin{figure*}[!t]
    \centering
    \includegraphics[width=0.8\textwidth]{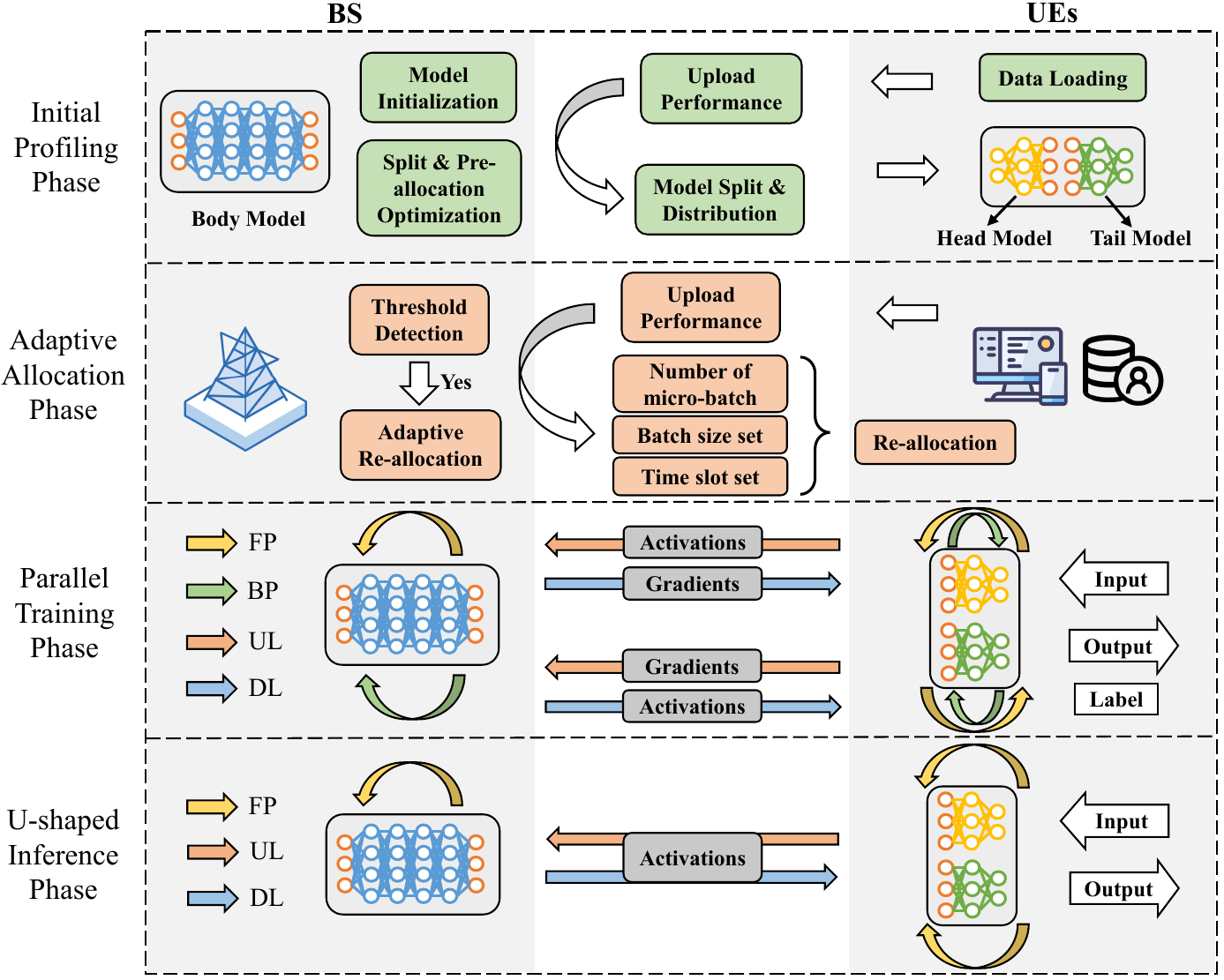}
    \caption{{System overview of AC$^2$P$^2$SL.}}
    \vspace{-0.2cm}
    \label{overview}
\end{figure*}

\subsection{Communication Model}\label{section2.2}

Note that the activation gradients generated during BP have the same dimension as the activation outputs produced during FP for the same network layer, we uniformly denote the activation output data size of $l$-th layer as $a_l$. Consequently, regarding the communication interface required between the UEs and the BS, the transfer loads at the two cut layers $l_1,l_2$ are expressed as $a_{l_1}$ and $a_{l_2}$, respectively.

To address the instability of UE performance in practical training scenarios, we consider a time division duplexing (TDD) communication system with time division multiple access (TDMA)\footnote{Both TDMA and FDMA are merely methods for modeling communication rates and do not affect overall AC$^2$P$^2$SL framework.}, which facilitates dynamic resource allocation to accommodate discontinuous, time-varying, and asynchronous transmission of UEs. Specifically, TDMA partitions time into periodic frames of length $T$, with each frame further subdivided into multiple time slots of length $\tau$. Each slot is assigned to a specific UE for data transmission, while the number of time slots allocated to the $i$-th UE is denoted by $s_i$. By dynamically scheduling resources in the time domain, this mechanism effectively reduces inter-user interference, thereby ensuring parallel and reliable data transmission. Consequently, the constraint relating the time frame and the allocated time slots is expressed as

{
\begin{equation}
\label{slotbound}
\sum_{i=1}^N\tau s_i \le T.
\end{equation}
}

It is noteworthy that $\rho$ is defined as the ratio of the number of uplink time slots to that of downlink time slots. We adopt the quasi-static block fading channels, holding that the UE channels remain approximately stable over the short duration of a single data batch training. This implies that the channel gain remains constant within each communication round, but may vary across different data batches. Accordingly, based on the Shannon's theorem, the achievable uplink and downlink transmission rates of the $i$-th UE and BS are formulated as

{\begin{equation}r_i^u=\frac{\tau s_i\rho}{T(1+\rho)}BWlog_2\left(1+\frac{G_iG_0p_i h_i}{BW N_0}\right),\end{equation}
\begin{equation}r_i^d=\frac{\tau s_i}{T(1+\rho)}BWlog_2\left(1+\frac{G_0G_ip_0 h_i}{BW N_0}\right),\end{equation}
where $BW$ represents the total system bandwidth shared by all UEs,} $p_i$ denotes the uplink transmit power of the $i$-th UE, and $p_0$ indicates the downlink transmit power of the BS. $G_i$ and $G_0$ stands for the antenna gain of the $i$-th UE and BS while $N_0$ denotes the power spectral density (PSD) of the noise. Finally, $h_i$ represents the channel gain for the $i$-th UE, which includes path loss and shadow fading, as well as multipath fading.

\subsection{Storage Model}\label{section2.3}

For UEs in the wireless edge network, the constrained storage resources necessitate imposing limits on the sizes of the partitioned head and tail models. Similarly to the memory access analysis in Section \ref{section2.1}, we categorize the memory usage of each layer into two components: the fixed model parameter usage $u_l$, and the variable intermediate activation usage per sample $\Delta u_l$. Consequently, taking into account both the head and tail models, the total memory footprint of UE $i$ is constrained by
\begin{equation}
    \sum_{l=1}^{l_1} (u_l + b_i \Delta u_l) + \sum_{l=l_2+1}^{L} (u_l + b_i \Delta u_l) \le U_i,
\end{equation}
where $U_i$ represents the maximum memory capacity of UE $i$. The storage constraint imposes a strict upper bound on both the split layer $l_1,l_2$ and the allowable batch size $b_i$. Conversely, for the BS, we assume it possesses sufficient storage resources for entire model training.

\section{AC$^2$P$^2$SL Framework and Training Workflow } \label{section3}
In this section, we present the system overview of the AC$^2$P$^2$SL framework across its various phases, along with the parallel training workflow.

\subsection{System Overview}\label{section3.1}
As illustrated in Fig.~\ref{overview}, The overall AC$^2$P$^2$SL framework comprises four primary phases: initial profiling, adaptive allocation, parallel training and U-shaped inference phase.

In the initial profiling phase, the local data on UEs is loaded while BS initializes the entire model. Subsequently, each UE uploads its communication and computation performance to the BS. Based on the parameters, the BS makes the SPA optimization which is presented in Section \ref{section4.2} to determine the near-optimal model split layer, number of micro-batches, batch size set, and allocated time slots set. Upon completing this optimization, the BS distributes the partitioned head and tail sub-model, along with the optimization results, to the UEs.

The adaptive allocation phase occurs before each training round to dynamically allocate computing workloads and communication resources. Specifically, the BS monitors the performance parameters uploaded by the UEs. If variations exceed a predefined threshold $\delta$, the ARA optimization presented in Section \ref{section4.3} concerning the number of micro-batches, batch size, and time slot is triggered to maximize operational efficiency. Furthermore, in scenarios where specific system nodes encounter failures and cease participation, this adaptive allocation functions as an elastic fault-tolerance mechanism, mitigating the adverse effects of UE heterogeneity and temporal variations.

In the training workflow, serial execution of computation and communication for a single batch inevitably introduces idle time across the various processing stages. Motivated by pipeline parallelism, we adopt a similar strategy. Specifically, UEs split a single batch uniformly into $k$ micro-batches, where each micro-batch consists of $b_i/k$ data samples. This partitioning enables UE to train a sequence of micro-batches continuously without increasing the memory footprint required for each training round. By utilizing the idle time of a single micro-batch, different stages concurrently process other micro-batch. Thus, AC$^2$P$^2$SL achieves micro-batch level parallelism that effectively overlaps computation with communication.

The U-shaped inference phases depicted in Fig.~\ref{overview} illustrate the USL process for a single UE-BS pair. In the subsequent inference phase, UE can execute inference tasks while keeping both input data and output results local without model parameter transmission, thus achieving comprehensive privacy protection. In practice, the proposed AC$^2$P$^2$SL framework supports parallel split inference across multiple UEs. By leveraging data batch partitioning and communication-computation pipeline parallelism in the same way, the framework achieves efficient edge inference. 

\begin{figure*}[!t]
    \centering
    \includegraphics[width=0.9\textwidth]{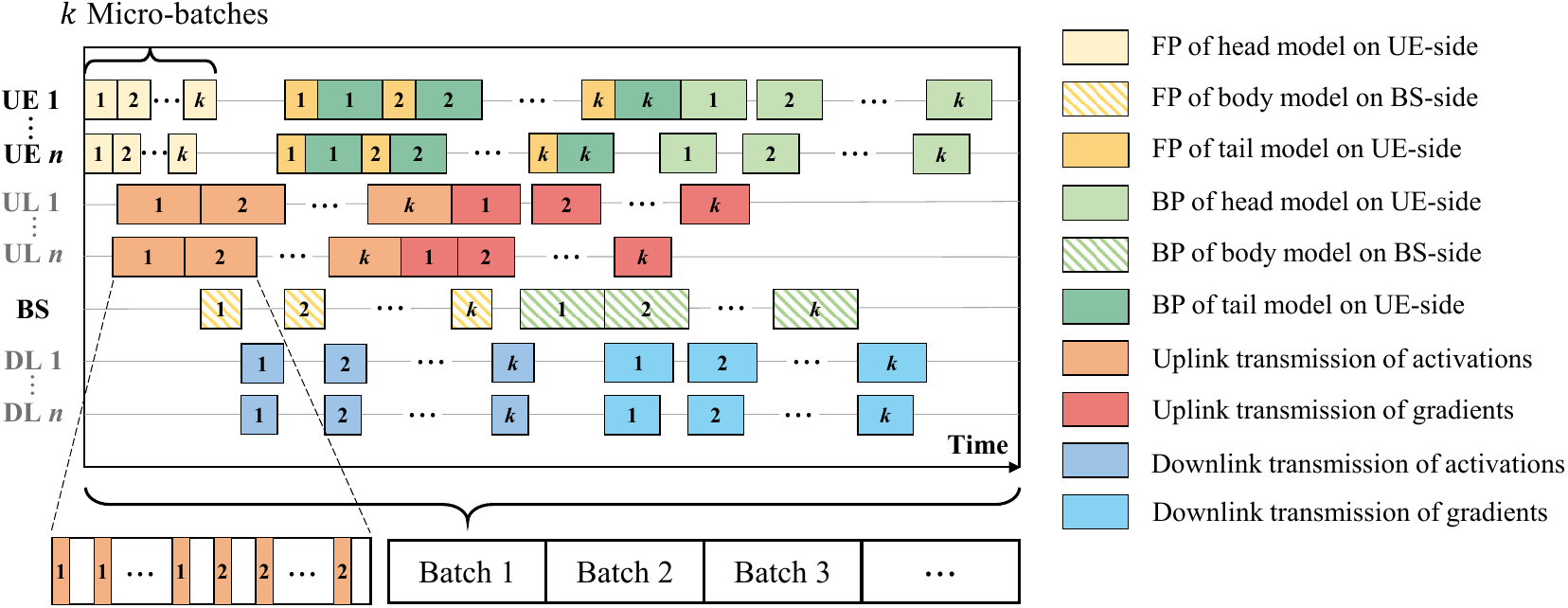}
    \caption{Training workflow of AC$^2$P$^2$SL. The uplink and downlink transmission block of different UEs consist of multiple non-overlapping time slots.}
    \label{workflow}
\end{figure*}

\subsection{Parallel Training Workflow} \label{section3.2}
Following the Initial profiling and adaptive resource allocation phases, the workflow of the U-shaped training phase is illustrated in Fig.~\ref{workflow}. For convenience, we describe the training round for a single data batch, which includes the parallel training of $k$ micro-batches. In the figure, the vertical axis sequentially represents four distinct physical stages: multi-UE computation, multi-UE uplink transmission, BS computation, and multi-UE downlink transmission, while the horizontal axis depicts the timeline as each data micro-batch goes through these stages. Next, we introduce the micro-stages that occur in the different physical stages:

\subsubsection{Head Model FP and Output Uplink Transmission}
The process begins with the parallel computation of UEs. Each UE concurrently inputs its $k$ local data micro-batches, denoted as $x_{i,j}$, into the head model for FP sequentially. These tasks are enqueued in the computation buffer. Through FP, the activation output of head model is generated as
{\begin{equation}\label{eq1}
a_h(i,j) = \mathcal{F}(W_h(i), x(i,j)).    
\end{equation}}
Upon completing FP of the current micro-batch, the activation output $a_h(i,j)$ is immediately added to the transmission queue, awaiting transmission during the allocated idle time slot in the uplink. Noting that due to heterogeneity among UEs, variations in computational and communication capabilities lead to distinct FP and uplink transmission durations, thereby inducing synchronization delays.

\subsubsection{Body Model FP and Output Downlink Transmission}
Once receiving the activation outputs of the $j$-th micro-batch from all UEs, the BS aggregates and concatenates them along the batch dimension, as 
\begin{equation}\label{eq2}
a_{h}(j) = [a_{h}(1,j); a_{h}(2,j); \dots; a_h(n,j)].    
\end{equation}
Subsequently, the BS inputs the aggregated output $a^{h}(j)$ to obtain the activation output of the body model by performing FP, denoted as 
{\begin{equation}\label{eq3}
a_b(j) = \mathcal{F}(W_b, a_h(j)).    
\end{equation}}
Later, following the logic of aggregation, the BS splits the output to derive $a_b(i,j) = a_b(j)[i]$. These partitioned outputs are then enqueued for downlink transmission to their respective UEs during available allocated time slots.

\subsubsection{Tail Model FP, BP and Gradient Uplink Transmission}
After receiving the corresponding output of body model, the $i$-th UE inputs it into the tail model for FP whenever the computation queue is idle, yielding the activation output 
{\begin{equation}\label{eq4}
a_t(i,j) = \mathcal{F}(W_t(i), a_b(i,j)).    
\end{equation}}
The local loss is then calculated using the corresponding data label $y(i,j)$, expressed as 
{\begin{equation}\label{eq5}
l(i,j) = \mathcal{L}(a_t(i,j), y(i,j)).    
\end{equation}}
Crucially, computation of the tail model adheres to the one forward step by one backward (1F1B) principle. Specifically, immediately after completing FP of the $j$-th micro-batch, the UE performs the homologous BP to derive the gradients for the tail model parameters $\Delta W_t(i,j)$ and the input data $g_b(i,j)$ (i.e., activations from the body model), denoted as 
{\begin{equation}\label{eq6}
\Delta W_t(i,j), g_b(i,j) = \mathcal{B}(l(i,j), W_t(i), a_b(i,j)).    
\end{equation}}
Theoretically, besides the variables explicitly denoted in the equations, BP needs to input the FP intermediate activation outputs of each layer in the sub-model. These activations constitute the computational graph retained in the device memory and are released  only on completion of BP. For the sake of notational brevity, we treat these retained activations as implicit default inputs and omit them from the mathematical formulations. Once completed, the input data gradients are added to the uplink transmission queue, pending transmission during the idle slots.

\subsubsection{Body Model BP and Gradient Downlink Transmission}
Following receipt of the activation gradients for micro-batch $j$ from all UEs, the BS aggregates them along the batch dimension, analogously to the forward output as 
\begin{equation}\label{eq7}
g_b(j) = [g_b(1,j); g_b(2,j); \dots ;g_b(n,j)].    
\end{equation}
Then BP is executed for the body model to obtain the gradients for the model parameters and the aggregated head model activations, represented as 
{\begin{equation}\label{eq8}
\Delta W_b(j), g_h(j) = \mathcal{B}(g_b(j), W_b, a_h(j)).    
\end{equation}}
Subsequently, the BS splits the activation gradients to obtain $g_h(i,j) = g_h(j)[i]$, enqueuing them for downlink transmission to the original UEs.

\begin{algorithm}[!t]

\small
\caption{The AC$^2$P$^2$SL Framework} \label{algo1}
\begin{algorithmic}[1]
\REQUIRE $x^m(i,j)$, $y^m(i,j)$, $\eta$, $M$
\ENSURE ${W_h^{M}}(i)$, ${W_t^{M}}(i)$, ${W_b^{M}}$. 
\STATE UEs upload capability parameters to BS 
\STATE BS get $\{l_1, l_2, k^{0}, b_i^{0}, s_i^{0}\}_{i \in \mathcal{N}}$ by Algorithm \ref{algo3}
\STATE Initialize model parameters and split model
\STATE BS distributes split model and allocated $(k^{0}, b_i^{0}, s_i^{0})$ to UEs
\FOR{$m=1,2,\ldots,M$}
\IF{${(t^m-t^{m-1})}/{t^{m-1}}>\delta$}
    \STATE BS get $\{k^{m}, b_i^{m}, s_i^{m}\}_{i \in \mathcal{N}}$ by Algorithm \ref{algo4}
\ELSE
    \STATE $\{k^{m}, b_i^{m}, s_i^{m}\}_{i \in \mathcal{N}} \leftarrow \{k^{m-1}, b_i^{m-1}, s_i^{m-1}\}_{i \in \mathcal{N}}$
\ENDIF
\item[] \texttt{// Runs on UEs}
    \FOR{\textbf{each} $i \in \mathcal{N}$ \textbf{in parallel}}
        \FOR{$j=1,2,\ldots,k^m$}
            \STATE $a_h^m(i,j) \leftarrow \mathcal{F}(W_h^m(i), x^m(i,j))$
            \STATE Transmit $a_h^m(i,j)$ to BS
        \ENDFOR
    \ENDFOR

    \IF{Receive $a_b^m(i,j)$ of $j$-th micro-batch from BS}
        \STATE $l^m(i,j) \leftarrow \mathcal{L}(\mathcal{F}(W_t^m(i), a_b^m(i,j)), y^m(i,j))$
        \STATE $\Delta W_t^m(i,j), g_b^m(i,j) \leftarrow \mathcal{B}(l^m(i,j), W_t^m(i), a_b^m(i,j))$
        \STATE Transmit $g_b^m(i,j)$ to BS
    \ENDIF

    \IF{Receive $g_h^m(i,j)$ of $j$-th micro-batch from BS}
        \STATE $\Delta W_h^m(i,j) \leftarrow \mathcal{B}(g_h^m(i,j), W_h^m(i))$
    \ENDIF

    \STATE Update head and tail model as (\ref{eq10})

\item[] \texttt{// Runs on BS}
    \IF{Receive $a_h^m(i,j)$ of $j$-th micro-batch from all UEs}
        \STATE $a_{h}^m(j) \leftarrow [a_{h}^m(1,j); a_{h}^m(2,j); \dots ; a_h^m(n,j)]$
        \STATE $a_b^m(j) \leftarrow \mathcal{F}(W_b^m, a_h^m(j))$
        \STATE Transmit split $a_b^m(i,j)$ to UEs
    \ENDIF

    \IF{Receive $g_b^m(i,j)$ of $j$-th micro-batch from all UEs}
        \STATE $g_{b}^m(j) \leftarrow [g_{b}^m(1,j); g_{b}^m(2,j); \dots ; g_b^m(n,j)]$
        \STATE $\Delta W_b^m(j), g_h^m(j) \leftarrow \mathcal{B}(g_b^m(j), W_b^m, a_h^m(j))$
        \STATE Transmit split $g_h^m(i,j)$ to UEs
    \ENDIF

    \STATE Update body model as (\ref{eq11})
\ENDFOR
\end{algorithmic}
\end{algorithm}

\subsubsection{Head Model BP and Parameters Update}
Finally, having received the corresponding activation gradients for the $j$-th micro-batch, the UE performs BP for the head model whenever the computation queue permits, yielding 
{\begin{equation}\label{eq9}
\Delta W_h(i,j) = \mathcal{B}(g_h(i,j), W_h(i)).    
\end{equation}
During the $m+1$-th round of batch training, after sequentially completing BP of all $k$ micro-batches, the head and tail model parameters of each UE are first updated using the mini-batch gradient descent (MBGD) method \cite{MBGD}, and then aggregated by FedAvg algorithm \cite{FedAvg}, as 
\begin{equation}\label{eq10}
   {W_{h/t}^{m+1}} = W_{h/t}^m - \frac{\eta}{B}\sum_{i=1}^N \frac{b_i}{k}\sum_{j=1}^k \Delta W_{h/t}^m(i, j), 
\end{equation}
where $\eta$ is the learning rate. Simultaneously, once the BS completes the BP for $k$ micro-batches and the computation queue is idle, it updates the body model parameters by using the MBGD method via 
\begin{equation}\label{eq11}
   {W_{b}}^{m+1} = W_b^m - \frac{\eta}{k} \sum_{j=1}^k \Delta W_b^m(j). 
\end{equation}}

{
\subsection{Convergence Analysis} \label{section3.3}
After all of the above processes, this concludes the training process for the current data batch, and the overall training framework is shown in Algorithm \ref{algo1}. Notably, while the synchronization of activations and gradients occurs strictly per-micro-batch across all users, the model parameters are updated synchronously only at the global batch boundary. Therefore, our pipeline mechanism does not introduce asynchronous gradient staleness. We have Theorem \ref{theorem1}, which theoretically guarantees that the convergence speed and final accuracy of Algorithm \ref{algo1} remain completely unaffected. 

\begin{theorem}\label{theorem1}
    Given the same global batch size, the parameter update process of AC$^2$P$^2$SL is mathematically identical to regular model training without batch splitting.
\end{theorem}

\begin{proof}
For the shared body sub-model of the BS, the cumulative parameter gradients of the $\frac{B}{k}$ samples in the $j$-th micro-batch is represented as
\begin{equation}
   \Delta W_b^m(j)=\frac{k}{B}\sum_{t=1}^\frac{B}{k} \Delta w_b^m(\frac{B(j-1)}{k}+t), 
\end{equation}
where $\Delta w_b^m(\frac{B(j-1)}{k}+t)$ denotes the parameter gradient of the $t$-th sample in the $j$-th micro-batch. Consequently, the average parameter gradient of k micro-batches is exactly equal to the total parameter gradient of all samples in the undivided global batch, simply as
\begin{equation}
   \frac{1}{k}\sum_{j=1}^k \Delta W_b^m(j) = \frac{1}{B} \sum_{s=1}^B \Delta w_b^m(s).
\end{equation}

Similarly, for the head and tail sub-models distributed across UEs, the parameter gradient generated by samples in the $j$-th micro-batch of UE $i$ can be expressed as
\begin{equation}
   \Delta W_{h/t}^m(i,j) = \frac{k}{b_i}\sum_{t=1}^\frac{b_i}{k} \Delta w_{h/t}^m(i, \frac{b_i(j-1)}{k}+t),
\end{equation}
where $\Delta w_{h/t}^m(i, \frac{b_i(j-1)}{k}+t)$ denotes the parameter gradient of the $t$-th sample in the $i$-th UE's $j$-th micro-batch. After completing BP of $k$ micro-batches, UE $i$ aggregates $k$ parameter gradients by FedAvg to obtain the total parameter gradient for its batch of samples, presented as
\begin{equation}
   \frac{1}{k}\sum_{j=1}^k \Delta W_{h/t}^m(i,j) = \frac{1}{b_i} \sum_{s=1}^{b_i} \Delta w_{h/t}^m(i, s).
\end{equation}
Then, by weighted averaging the parameter gradients of all UEs, the aggregated gradients of head and tail sub-models to be updated can be denoted as
\begin{equation}
   \frac{b_i}{B}\sum_{i=1}^N \frac{1}{k}\sum_{j=1}^k \Delta W_{h/t}^m(i, j) = \frac{1}{B} \sum_{s=1}^B \Delta w_{h/t}^m(s),
\end{equation}
which is also equal to the total parameter gradient of all samples in the undivided global batch. Obviously, this parallel mechanism operates entirely independently of the parameter update procedures for the head, body, and tail sub-models. Consequently, the AC$^2$P$^2$SL framework is highly compatible with various federated learning model aggregation algorithms, although we default to the widely adopted FedAvg to aggregate model parameters on UE-side. 
As a result, Theorem \ref{theorem1} is proved.

\end{proof}

}

\section{Optimization for Split and Allocation} \label{section4}
In this section, we establish a dynamic programming model to analyze training time and make optimizations on split pre-allocation and adaptive re-allocation.

\subsection{Time Analysis} \label{section4.1}
Considering that the whole training process for a single data batch involves a complex pipelining of $k$ micro-batch tasks across $n$ parallel UEs, spanning 4 distinct physical queues and 9 logical micro-stages (where FP and BP of tail model are treated as a unified entity), the system is characterized by intra-stage resource re-entry, pipeline blocking, and inter-stage synchronization barriers. Consequently, the total training time $t$ cannot be accurately formulated using a simple closed-form equation. In particular, since the computational overhead associated with loss calculation and parameter updating is negligible, these processes are omitted from the calculation of the total computation time.

To address this, we formulate a dynamic programming model to precisely calculate the total time of the parallel process. This model tracks the progression of $k$ micro-batches for $n$ UEs across the 9 micro-stages while simultaneously managing resource availability across the 4 physical queues. Let $C(i, j, s)$ denote the completion timestamp (measured from the start of the training phase) of the $s$-th micro-stage for the $j$-th micro-batch of UE $i$, where $i \in \{1, \dots, n\}$, $j \in \{1, \dots, k\}$, and $s \in \{1, \dots, 9\}$.

Firstly, the process commences with head model FP for the $j$-th micro-batch of UE $i$. As each micro-batch is sequentially input the head model, and in accordance with the analysis presented in Section \ref{section2.2}, the computational capability of UE $i$ is denoted as $f_i(b_i/k, 1, l_1)$ and the total computational workload is calculated as $\frac{b_i}{k} \sum_{l=1}^{l_1} c_l^f$. Consequently, the completion time for the local computation queue is formulated as 
\begin{equation}
    C(i,j,1) = j \cdot \frac{b_i \sum\limits_{l=1}^{l_1} c_l^f}{k f_i^f({b_i}/{k}, 1, l_1)}.
\end{equation}

Due to the queuing constraints at the UE, the computation for a subsequent micro-batch is blocked until the preceding one completes. Building upon this, the uplink transmission for the $j$-th micro-batch is contingent upon two synchronization constraints: the completion of its own FP calculation and the completion of the transmission for the preceding micro-batch. With a transmission payload of $\frac{b_i}{k} a_{l_1}$, the completion time for the uplink transmission queue is expressed as
\begin{equation}
    C(i, j, 2) = \max \left\{ C(i, j, 1), C(i, j-1, 2) \right\} + \frac{b_i a_{l_1}}{k r_i^u}.
\end{equation}

Subsequently, BS aggregates the outputs of the same micro-batch from all UEs, necessitating a synchronization barrier. Thus, body model FP start time of the $j$-th micro-batch is determined by two values: the last uplink transmission across all UEs and the FP the previous $(j-1)$-th micro-batch. Given that the aggregated input volume is the sum of batch sizes from $n$ UEs, the BS computation queue is formulated as 
{\begin{equation}
    C(j, 3)\! =\! \max_i \left\{C(i, j, 2), C(j-1, 3) \right\} \!+\! \frac{B \sum\limits_{l=l_1+1}^{l_2} c_l^f}{k f_0^f(\frac{B}{k}, l_1\!+\!1, l_2)}.\!
\end{equation}}

The BS then distributes the partitioned outputs to each UE via the downlink. Since the preceding FP stage is a synchronized process, all downlink transmissions begin with the completion of that stage. With a transmission load of $\frac{b_i}{k} a_{l_2}$, the downlink transmission queue is defined as 
\begin{equation}
    C(i, j, 4) = \max \left\{ C(j, 3), C(i, j-1, 4) \right\} + \frac{b_i a_{l_2}}{k r_i^d}.
\end{equation}

Since the start times of the first four micro-stages correspond to the idle states of the respective physical queues, their initial values are set to zero as $C(i, 0, 1) = C(i, 0, 2) = C(i, 0, 3) = C(i, 0, 4) = 0$. However, the subsequent micro-stages involve resource re-entry into these physical queues. So, their initial values are determined by the completion times of the last micro-batch $k$ from the preceding logical micro-stage that occupied the same physical resource queue. 

Upon receiving the downlink output, each UE performs FP and BP for the tail model on micro-batch $j$ under the 1F1B principle, which we treat as a unified micro-stage. However, distinguishing the computational and memory access differences between FP and BP, we calculate their durations sequentially. The UE computation queue is represented as 
\begin{equation}
\begin{split}
    C(i, j, 5) &= \max \{ C(i, j, 4), C(i, j-1, 5) \} + \\ & \quad \frac{b_i \sum_{l=l_2+1}^{L} c_l^f}{k f_i^f({b_i}/{k}, l_2\!+\!1, L)} + \frac{b_i \sum_{l=l_2+1}^{L} c_l^b}{k f_i^b({b_i}/{k}, l_2\!+\!1, L)},
\end{split}
\end{equation}
where the initial value $C(i, 0, 5) = C(i, k, 1)$. Following the tail model BP, the UE transmits the gradients via the uplink. With a payload size of $a_{l_2}$, the uplink transmission queue is formulated as 
\begin{equation}
    C(i, j, 6) = \max \left\{ C(i, j, 5), C(i, j-1, 6) \right\} + \frac{b_i a_{l_2}}{k r_i^u},
\end{equation}
with initialization $C(i, 0, 6) = C(i, k, 2)$. BS then receives gradients from all UEs and performs batch aggregation. Similarly to body model FP, this involves synchronization across all UEs. The BS computation queue is expressed as 
{\begin{equation}
    C(j, 7)\! =\! \max_i \left\{C(i, j, 6), C(j-1, 7) \right\} \!+\! \frac{B \sum\limits_{l=l_1+1}^{l_2} c_l^b}{k f_0^b(\frac{B}{k}, l_1\!+\!1, l_2)},
\end{equation}
with initialization $C(0, 7) = C(k, 3)$.} Subsequently, BS transmits the partitioned gradients to the UEs. With a payload of $a_{l_1}$, the downlink transmission queue is defined as 
\begin{equation}
    C(i, j, 8) = \max \left\{ C(j, 7), C(i, j-1, 8) \right\} + \frac{b_i a_{l_1}}{k r_i^d},
\end{equation}
with the initial value $C(i, 0, 8) = C(i, k, 4)$. Finally, upon receiving the downlink gradients, the UE executes BP of head model. The computation queue is given by 
\begin{equation}
    \!C(i, j, 9)\! =\! \max\! \left\{C(i, j, 8\!), \!C(i, \!j-1,\! 9)\! \right\} \!+\! \frac{b_i \sum\limits_{l=1}^{l_1} c_l^b}{k f_i^b(\frac{b_i}{k}, 1, l_1)},\!
\end{equation}
with initialization $C(i, 0, 9) = C(i, k, 5)$. Consequently, the total training time of a single data batch under the AC$^2$P$^2$SL framework is derived as 
\begin{equation}\label{objective}
    t = \max_i C(i, k, 9).
\end{equation}

\subsection{Split and Pre-allocation Optimization} \label{section4.2}
{
In the proposed framework, total pipeline training time is fundamentally governed by a complex interplay of layer split, pipeline depth, and heterogeneous resource allocation. To minimize this overall latency, it is imperative to jointly optimize these key variables.

\subsubsection{Problem Formulation}
In USL, split layer $l_1$ and $l_2$ directly dictate the number of layers assigned to the head, body, and tail sub-models, satisfying the feasibility boundary $1<l_1<l_2<L$. This partition fundamentally governs the computational workload at both UEs and BS sides, as well as the data volume transmitted over uplink and downlink channels, thereby determining the training time. Crucially, because the computational intensity and output tensor dimensions of each layer vary dramatically across different neural network architectures, the selection of these split points must be adaptively tailored to the fine-grained layer-profiling parameters of the specific model.

Simultaneously, the number of micro-batches $k$ per training round determines the attainable hardware computing performance of devices processing each micro-batch and the pipeline depth of parallelism among micro-stages. Specifically, a larger $k$ enhances the concurrency of the pipeline, which helps compress idle time while simultaneously reducing the sample workload contained within each individual micro-batch. This lower operational intensity degrades hardware compute utilization, thereby increasing the computational duration. Furthermore, $k$ is naturally upper-bounded by the minimum batch size allocated across the devices, requiring $1 \le k \le \min_i b_i$, $k \in \mathbb{N}$. Navigating this trade-off to select appropriate micro-batch number is therefore vital for maximizing pipeline efficiency.

Given a total batch size, the batch size allocated to each UE $b_i$ determines the number of training samples in its single round. Because both the computational workload and the communication volume are positively correlated with $b_i$, optimizing this allocation is essential given the highly heterogeneous hardware capabilities of the UEs. 
Additionally, disparities in transmit power and path loss during wireless propagation introduce significant variations in communication durations among UEs. Within a TDMA system, the number of time slots allocated to each UE $s_i$ maintains a positive correlation with its transmission speed. 
By optimizing $b_i, s_i$ to tune computation and communication time of UEs, the BS can balance the execution timelines of parallel UE micro-stages, thereby minimizing the synchronization waiting time caused by straggler effects at the global aggregation barrier, subject to the total batch bound in (\ref{batchbound}) and the bounded frame duration in (\ref{slotbound}). 

Therefore, in the initial profiling phase, the BS collects specific parameters uploaded by the UEs, including transmit power, location information, FLOPS, etc. Leveraging these heterogeneous parameters, the BS conducts joint optimization over model split layer $l_1, l_2$, number of micro-batches $k$, batch size set $\boldsymbol{b}$, and time slot set $\boldsymbol{s}$. We perform quantitative analysis and modeling based on the constraints of these variables with the objective of minimizing the total pipeline training time. The joint optimization problem is formulated as follows.
}
\begin{align} \label{P1}
P1: \quad & \min_{l_1, l_2, k, \boldsymbol{b}, \boldsymbol{s}} t(l_1, l_2, k, \boldsymbol{b}, \boldsymbol{s})\\
\text{s.t.} \quad & \text{C1: } 1 \le l_1 < l_2 \le L-1, \quad l_1, l_2 \in \mathbb{N},\nonumber\\
& \text{C2: } \sum_{l=1}^{l_1} (u_l + b_i \Delta u_l) + \sum_{l=l_2+1}^{L} (u_l + b_i \Delta u_l) \le U_i,\nonumber\\
& \text{C3: } 1 \le k\le \min_i b_i, \quad k \in \mathbb{N},\nonumber\\
& \text{C4: } \sum_{i=1}^{n} b_i = B, \quad b_i \in \mathbb{N},\nonumber\\
& \text{C5: } \sum_{i=1}^{n} \tau s_i \le T, \quad s_i \in \mathbb{N}, \nonumber
\end{align}

{ The optimization objective is to minimize the training time for a single data batch shown in Section \ref{section4.1}. Due to the complexity of 9-micro-stage pipeline, sharing of physical resources among certain micro-stages, and presence of inter-stage blocking, the objective function exhibits high non-linearity. Furthermore, constraint C1 defines the boundary conditions for the cut layers within the two-layer splitting strategy. Constraint C2 addresses the varying storage capabilities among UEs. It ensures that the aggregate memory footprint, including the head and tail model parameters as well as intermediate activations, is strictly bounded by each device's maximum memory usage. Constraint C3 ensures that the number of pipeline micro-batches does not exceed the local batch size allocated to the UE. Constraint C4 dictates that the sum of data batches allocated across all users exactly matches the total global batch size. Lastly, constraint C5 imposes an upper bound on the total number of communication time slots allocated to the UEs.}

Since the decision variables $(l_1, l_2, k, \boldsymbol{b}, \boldsymbol{s})$ are all integers and the objective function and constraints are nonlinear and coupled, $P1$ constitutes a mixed-integer nonlinear programming (MINLP) problem. Furthermore, the non-convexity and discontinuity of the objective function preclude the direct derivation of a global optimal solution. Therefore, we propose a SPA algorithm based on hierarchical decomposition and alternating optimization (AO)\cite{AO}. This approach decomposes $P1$ into two loops: the outer loop is responsible for the feasibility pruning and search of the split layer $l_1$ and $l_2$, while the inner loop solves the joint resource allocation sub-problem regarding $(k, \boldsymbol{b}, \boldsymbol{s})$ via AO. { The overall algorithm is summarized in Algorithm \ref{algo3}.}

\begin{algorithm}[t]

\small
\caption{SPA Algorithm Based on AO} \label{algo3}
\begin{algorithmic}[1]
\REQUIRE FLOPS, memory access, and storage parameters of each model layer, convergence tolerance $\epsilon$.
\ENSURE $l_1^*, l_2^*, k^*, \boldsymbol{b}^*,\boldsymbol{s}^*$.
\STATE Initialization: $k^{(0)}, \boldsymbol{b}^{(0)}, \boldsymbol{s}^{(0)}, t_{min}$.
\FOR{$l_1=1,2,\ldots,L-2$}
    \FOR{$l_2=L-1,L-2,\ldots,l_1$}
        \STATE Select $(l_1,l_2)$ by the outer loop pruning search as Theorem \ref{theorem2}.
        \STATE Obtain $(k, \boldsymbol{b}, \boldsymbol{s})$ by the inner loop AO as Algorithm \ref{algo4} inputting $\epsilon$.
        \IF{$t(l_1,l_2,k,\boldsymbol{b},\boldsymbol{s}) < t_{min}$}
        \STATE $t_{\min} \gets (l_1,l_2,k,\boldsymbol{b},\boldsymbol{s}), l_1^* \gets l_1, l_2^* \gets l_2, k^* \gets k,$\\$ \boldsymbol{b}^* \gets \boldsymbol{b}, \boldsymbol{s}^* \gets \boldsymbol{s}$.
        \ENDIF
    \ENDFOR
\ENDFOR
\end{algorithmic}
\end{algorithm}

\subsubsection{Outer Loop Pruning Search}
Due to the variance in computational and communication overhead across different network layers, the objective function is discrete and highly non-convex with respect to $l$. Given the finite layer number, we employ a double reverse loop traversal over all available $l_1, l_2$ under constraint C1. This global search in the outer loop guarantees global optimality for the discrete structure. To reduce the computational complexity of the inner loop, we have the following simple observation: 

\begin{theorem}\label{theorem2}
    For the given $l_1,l_2$ in (\ref{P1}), the maximum batch size allocated to the $i$-th UE is as follows: $b_i^{limit}(l_1, l_2) = \lfloor\frac{(U_i - M_{fix}(l_1, l_2))}{M_{var}(l_1, l_2)}\rfloor$, where $M_{fix}(l_1, l_2) = \sum_{l=1}^{l_1} u_l + \sum_{l=l_2+1}^{L} u_l$ and $M_{var}(l_1, l_2) = \sum_{l=1}^{l_1} \Delta u_l + \sum_{l=l_2+1}^{L} \Delta u_l$.
\end{theorem}

\begin{proof}
    Starting from the constraint C2 in (\ref{P1}), after some inequality manipulation, it is easy to arrive that  
    \begin{equation}
        b_i \le \frac{U_i-(\sum_{l=1}^{l_1}u_l+\sum_{l_2+1}^{L}u_l)}{\sum_{l=1}^{l_1}\Delta u_l+\sum_{l_2+1}^{L}\Delta u_l}.
    \end{equation}
    Considering that $b_i$ needs to be an integer, we thus obtain the above maximum $b_i^{limit}$.
\end{proof}

Based on Theorem \ref{theorem2}, we design a two-level pruning strategy. First, a single-node feasibility check is performed: If the memory limit of any specific UE satisfies 
\begin{equation}
    U_i < M_{fix}(l_1, l_2),
\end{equation}
the current split strategy is infeasible. If passed, a global capacity check is conducted then: If the sum of the maximum batch size of all UEs is less than the total volume, as 
\begin{equation}
    {\sum_{i=1}^{N} b_i^{limit}(l_1, l_2) < B,}
\end{equation}
it indicates that the current split pair $(l_1, l_2)$ cannot fulfill the task requirements, and thus the combination is directly eliminated from the search space. Only the combinations that pass these pruning tests are passed to the inner loop to calculate the corresponding minimum time.

\begin{algorithm}[t]

\small
\caption{{Inner Loop AO Algorithm}} \label{algo4}
\begin{algorithmic}[1]
\REQUIRE Convergence tolerance $\epsilon$, iteration index $r=0$.
\ENSURE $k^*, \boldsymbol{b}^*,\boldsymbol{s}^*$.
\STATE Initialization: $k^{(0)}, \boldsymbol{b}^{(0)}, \boldsymbol{s}^{(0)}, t_{min}$.
\REPEAT
    \STATE $r \gets r+1$
    \FOR {$k=1, 2, \ldots, \min\limits_i{b_i^{(r-1)}}$ with fixed $\boldsymbol{b}^{(r-1)}, \boldsymbol{s}^{(r-1)}$}
        \IF {$t(k) < t_{\min}$}
            \STATE $t_{\min} \gets t(k), k^{(m)} \gets k$.
        \ENDIF
    \ENDFOR
    \STATE Update the relaxing continuous variable $\tilde{\boldsymbol{b}}^{(r)}$ by PGD with fixed $\boldsymbol{s}^{r-1}$; 
    \STATE Rounding $\tilde{\boldsymbol{b}}^{(r)}$ into $\boldsymbol{b}^{(r)}$ by Algorithm \ref{algo2}.
    \STATE Update the relaxing continuous variable $\tilde{\boldsymbol{s}}^{(r)}$ by CVXPY with fixed $\boldsymbol{b}^{r}$;
    \STATE Rounding $\tilde{\boldsymbol{s}}^{(r)}$ into $\boldsymbol{s}^{(r)}$ by Algorithm \ref{algo2}.
\UNTIL{$|t(k^{(r)}, \boldsymbol{b}^{(r)}, \boldsymbol{s}^{(r)})-t(k^{(r-1)}, \boldsymbol{b}^{(r-1)},$\\$ \boldsymbol{s}^{(r-1)})| \leq\epsilon$.}
\IF{$t(k,\boldsymbol{b},\boldsymbol{s}) < t_{min}$}
\STATE $t_{\min} \gets (k,\boldsymbol{b},\boldsymbol{s}), k^* \gets k, \boldsymbol{b}^* \gets \boldsymbol{b}, \boldsymbol{s}^* \gets \boldsymbol{s}$.
\ENDIF
\end{algorithmic}
\end{algorithm}

\subsubsection{Inner Loop AO}
Given $(l_1, l_2)$, P1 is transformed into a resource allocation sub-problem of pipeline depth $k$, time slot set $\boldsymbol{s}=\{s_i\}$, and batch size set $\boldsymbol{b}=\{b_i\}$. Due to the high coupling among variables in the objective function, we employ AO to solve them iteratively. First, we fix $(\boldsymbol{b}, \boldsymbol{s})$. An excessively small $k$ results in low parallelism, while an excessively large $k$ degrades computational capability, both leading to increased total time. This causes the pipeline time to exhibit a convex-like trend with $k$. However, the discontinuous nature of the objective function prevents a direct explicit solution. Therefore, we perform discrete search for $k$ within $[1, \min_i b_i]$ until a local optimum is identified. { After fixing $k$, we can obtain Theorem \ref{theorem3}.

\begin{theorem}\label{theorem3}
During the inner-loop AO with fixed $k$, the continuous relaxations of the isolated sub-problems for $\tilde{\boldsymbol{b}} \in \mathbb{R}^n$ and $\tilde{\boldsymbol{s}} \in \mathbb{R}^n$ are strictly convex.
\end{theorem}

\begin{proof}
We first analyze the nonlinearity of the computing performance $f_i^{(f/b)}(\tilde{b}_i/k)$ as a function of $\tilde{b}_i$. Although the Roofline model introduces a step-function, the objective function maintains smoothness within the practical constraint range. According to the algebraic transformation of the computation model in (\ref{roofline}), the computation time function $T_{comp}(\tilde{b}_i)$ transforms into the pointwise maximum of linear functions, strictly preserving convexity as
\begin{equation}
    T_{comp}(\tilde{b}_i) = \max \left( \frac{C}{k F_{i}} \tilde{b}_i, \; \frac{k M + \tilde{b}_i \Delta M)}{k \beta_i} \right),
\end{equation}
where $M$ and $\Delta M$ represent fixed and variable memory access, $C$ denotes the computational workload. Meanwhile, the communication time function can be simplified as $T_{comm}(\tilde{b}_i) \propto \tilde{b}_i$. Obviously, it is a strictly convex function with respect to $\tilde{b}_i$.
Similarly, considering that the transmission rate is proportional to $\tilde{s}_i$, the communication time is simplified as $T_{comm}(\tilde{s}_i) \propto \tilde{s}_i^{-1}$ yields a positive second derivative, ensuring strict convexity. 
The objective function (\ref{objective}) composed of these two time function terms is constructed recursively using only two mathematical operators: the non-negative addition of stage delays and the pointwise maximum across UEs and previous stages. According to convex analysis, if the base time functions are convex, their non-negative sums and pointwise maximums strictly preserve convexity. 
As a result, Theorem \ref{theorem3} is proved.
\end{proof}

}

Next, we fix $(k, \boldsymbol{s})$, transforming $P1$ into a sub-problem of the batch size set. To handle the integer constraints, we strictly relax $b_i$ into continuous variables $\tilde{b}_i$. By simplifying C2 and C3, we derive the sub-problem $P2$ of $\tilde{b}_i$ as follows. { According to Theorem \ref{theorem3}, $P2$ is the minimization of a convex max function with box and equality constraints. We use projected gradient descent (PGD) method to solve it, ensuring constraint C6 is satisfied in each iteration.} Upon obtaining the optimized results $\tilde{b}_i$, we apply the largest remainder rounding algorithm for integer rounding, as detailed in Algorithm \ref{algo2}.
\begin{align} \label{P2}
P2: \quad & \min_{\boldsymbol{\tilde{b}}} \max_iC(i,k,9|\boldsymbol{\tilde{b}})\\
\text{s.t.} \quad & \text{C4: } \sum_{i=1}^{n} \tilde{b}_i = B,\nonumber\\
& \text{C6: } k \le \tilde{b}_i\le b_i^{limit}.\nonumber
\end{align}

\begin{algorithm}[t]
\small
\caption{Largest Remainder Rounding Algorithm} \label{algo2}
\begin{algorithmic}[1]
\REQUIRE {Continuous variable set $\boldsymbol{\tilde{x}}=\{\tilde{x}_i\}$, total number $N$.}
\ENSURE Rounded variable set $\boldsymbol{x}={x_i}$. 
\FOR{{$i=1,2,\ldots,N$}}
    \STATE $\Delta x_i = \tilde{x}_i - \lfloor \tilde{x}_i \rfloor$
\ENDFOR
\STATE {$\Delta x = \sum_{i=1}^N \tilde{x}_i - \sum_{i=1}^N \lfloor \tilde{x}_i \rfloor$}
\STATE Arrange $\Delta x_i$ in descending order
\FOR{{$i=1,2,\ldots,N$}}
\IF{$\Delta x_i$ is the top $\Delta x$ in $\{\Delta x_i\}$}
    \STATE $\Delta x_i \gets 1$
\ELSE
    \STATE $\Delta x_i \gets 0$
\ENDIF
\STATE $x_i=\lfloor \tilde{x}_i \rfloor + \Delta x_i$
\ENDFOR
\end{algorithmic}
\end{algorithm}

Finally, fixing $(k, \boldsymbol{b})$ and relaxing $s_i$ into continuous variables $\tilde{s}_i$, $P1$ becomes a sub-problem $P3$ of the time slot set, which is also a min-max problem under inequality constraints. { As in Theorem \ref{theorem3}, $P3$ is strictly convex with respect to $\boldsymbol{s}$.} Consequently, efficient convex optimization toolkits such as CVXPY\cite{CVX} can be employed. Similarly, we apply Algorithm \ref{algo2} to round $\tilde{s}_i$ into $s_i$.
\begin{align} \label{P3}
P3: \quad & \min_{\boldsymbol{\tilde{s}}} \max_iC(i,k,9|\boldsymbol{\tilde{s}})\\
\text{s.t.} \quad & \text{C5: } \sum_{i=1}^{n} \tilde{s}_i \le \frac{T}{\tau}. \nonumber
\end{align}

\subsubsection{{Convergence and Complexity Analysis}}
{
Bounded by the whole model layer $L$ and the pruning strategy, the outer architectural variables $(l_1, l_2)$ reside in a finite discrete combinatorial set so that the outer search strictly terminates in at most $\gamma$ outer steps, where $\gamma$ is the number of layers under Theorem \ref{theorem2}. Thus, the global convergence solely depends on the termination of the inner AO loop. We have Theorem \ref{theorem4} to solve this.
\begin{theorem}\label{theorem4}
    Under fixed $(l_1, l_2)$, Algorithm \ref{algo4} globally converges to a stable stationary configuration $\{k^*, \boldsymbol{b}^*, \boldsymbol{s}^*\}$ within a finite number of iterations.
\end{theorem}
\begin{proof}
For any fixed $(l_1, l_2)$, the inner loop variables $(k, \boldsymbol{b}, \boldsymbol{s})$ are strictly constrained by C3, C4 and C5. Therefore, the global feasible state space $\boldsymbol{\Omega} = \{k, \boldsymbol{b}, \boldsymbol{s}\}$ is a strictly bounded and finite discrete set. At iteration $m$, the continuous AO guarantees descent in Theorem \ref{theorem3}, but Algorithm \ref{algo2} introduces bounded localized perturbation, which breaks the guarantee of strict monotonic descent across successive evaluations $t^{(r)}$. To mathematically neutralize this limit-cycle oscillation, Algorithm \ref{algo4} maintains a global historical minimum, represented as
\begin{equation}
    t_{\min}^{(r)} = \min \left( t_{\min}^{(r-1)}, \; t(k^{(r)}, \boldsymbol{b}^{(r)}, \boldsymbol{s}^{(r)}) \right) \le t_{\min}^{(r-1)},
\end{equation}
which is strictly monotonically non-increasing. The physical pipeline execution time is strictly bounded below by the physical boundaries and hardware compute ceilings as $t_{\min}^{(m)} \ge T_{bound} > 0$. According to the monotone convergence theorem, any sequence of real numbers that is monotonically non-increasing and bounded from below must strictly converge to a finite limit. Due to the convergence of $\{t_{\min}^{(m)}\}$, and the algorithm navigates exclusively through the finite discrete state space $\boldsymbol{\Omega}$, the sequence cannot take on infinitely many distinct values. The inner loop enforces a truncation tolerance $|t^{(r)} - t^{(r-1)}| \le \epsilon$. Thus, the iterative sequence must terminate after a finite number of state transitions, stably collapsing onto the discrete stationary configuration yielding $t_{\min}$. 
As a result, Theorem \ref{theorem4} is proved.
\end{proof}
}

In the inner AO process, we initialize the variables and sequentially solve the sub-problems until the convergence criterion $\epsilon$ is met. The complexity of Algorithm \ref{algo3} comprises three main parts. Theoretically, the complexity of the outer loop is $O(L^2)$. But the pruning strategy in Theorem \ref{theorem2} significantly reduces the computational load. By leveraging the UE memory constraints, the number of layers requiring search is compressed to $\gamma$, resulting in a complexity of $O(\gamma^2)$. For the inner AO, the computational complexity is dominated by solving P2 and P3. For convex optimization problems with $n$-dimensional variables, the time complexity typically scales exponentially with the number of UEs. {Therefore, the total algorithm complexity can be expressed as $O(\gamma^2 \log(\frac{1}{\epsilon})\cdot N^3)$.}

\subsection{Adaptive Re-allocation Optimization} \label{section4.3}
Following the initial profiling phase, where the model split layers ($l_1, l_2$) are determined via the aforementioned SPA optimization, the BS broadcasts the initial parameters of the head and tail models to all UEs. Simultaneously, it distributes the allocated batch size $b_i$ and the number of time slots $s_i$ to each respective UE. While the system is expected to operate according to these optimized settings, the inherent time-varying performance of edge devices poses significant challenges during subsequent training rounds. Specifically, the computational capabilities of UEs may fluctuate due to sudden local processes or node failures, and communication rates may vary due to device mobility. These factors can cause the actual parallel training process to deviate from the initial design, necessitating the fault-tolerance mechanism for adaptive allocation.

Considering the limitations on UE performance variations, coupled with the strict requirements for model parameter privacy and transmission overhead, dynamically adjusting model split layers during training is prohibitively expensive. Consequently, we only target the tunable parameters on the UE side, considering $k, b_i, s_i$ as optimization variables for this phase. To this end, a threshold detection mechanism is established before each training round. If the deviation ratio of training time compared to the previous round exceeds a predefined fault-tolerance threshold as 
\begin{equation}
    \frac{t^m-t^{m-1}}{t^{m-1}} > \delta,
\end{equation}
the system deems that a significant performance fluctuation of specific UE has occurred. This triggers ARA optimization at BS before next training round begins, based on the latest parameters information uploaded by UEs and the re-optimized variables are then distributed to them. This optimization problem is formulated as $P4$.
\begin{align} \label{P4}
P4: \quad & \min_{k, \boldsymbol{b}, \boldsymbol{s}} t(k, \boldsymbol{b}, \boldsymbol{s})\\
\text{s.t.} \quad 
& \text{C2: } \sum_{l=1}^{l_1} (u_l + b_i \Delta u_l) + \sum_{l=l_2}^{L} (u_l + b_i \Delta u_l) \le U_i,\nonumber\\
& \text{C3: } 1 \le k\le \min_i b_i, \quad k \in \mathbb{N},\nonumber\\
& \text{C4: } \sum_{i=1}^{n} b_i = B, \quad b_i \in \mathbb{N},\nonumber\\
& \text{C5: } \sum_{i=1}^{n} \tau s_i \le T, \quad s_i \in \mathbb{N}, \nonumber
\end{align}

Since $P4$ constitutes a sub-problem of the original problem $P1$ with $l_1, l_2$ fixed, the ARA algorithm is easily obtained based on AO. Drawing upon the inner loop AO described in Section \ref{section4.2}, we decompose $P4$ into three sub-problems and solve them iteratively by fixing the other variables in the same way. { Consequently, we can also use Algorithm \ref{algo4} to solve this optimization.}

\section{Experimental Results} \label{section5}
In this section, we introduce the experimental parameter settings and present a variety of experimental results to comprehensively evaluate the performance of AC$^2$P$^2$SL. 

\subsection{Experiment Settings} \label{section5.1}
We consider a cellular network with a 500 \text{m} radius, where the BS is centrally located and UEs are randomly distributed. We adopt the line-of-sight (LoS) channel model described in \cite{Los}, characterized by an average path loss exponent of 2.1 and a shadow fading standard deviation of 3.6 \text{dB}. To emulate the heterogeneity performance, the peak FLOPS of each UE is uniformly selected from [100, 200] \text{GFLOPs}. Similarly, the memory bandwidth and maximum storage capacity follow random distributions within [5, 10] \text{GB/s} and [1, 2] \text{GB}. The remaining system parameters are specified in Table \ref{system}.

{ To perform reasonable evaluations, we conduct the common image classification task using the ImageNet-100 dataset, which consists of 100 varied categories selected from the ImageNet dataset \cite{Imagenet}. It contains 129,395 training samples and 5,000 test samples, and each data sample has dimensions of $224 \times 224 \times 3$. Experiments are conducted under both independent and identically distributed (IID) and non-IID conditions. For IID data setting, we randomly shuffle the training dataset and evenly distribute it among all UEs. Under non-IID setting, we use the Dirichlet distribution to distribute the training dataset, with the Dirichlet parameter set to 0.5 to represent the heterogeneity of local terminal data. Meanwhile, we randomly shuffle the test dataset and distribute it evenly among all UEs.

The models adopted for training include ResNet\cite{RESNET} series models and Vision Transformer (ViT)\cite{VIT}.} To determine specific model metrics prior to training, we calculate the FP and BP computational workloads for each network layer on a per-sample basis. Moreover, the fixed memory access volume of FP is represented by the model parameters, while for BP, it corresponds to the sum of model parameters and their gradients. The variable memory access volume is characterized by the data flow: for FP, it involves reading input data and writing output data; for BP, it involves reading the input data and output gradients and writing the input gradients. Regarding memory usage, the fixed component comprises model parameters, parameter gradients, and optimizer states, whereas the variable component consists of the size of intermediate activations generated during FP. 

\begin{table}[t]
\small
\renewcommand{\arraystretch}{1.3}
\centering 
\footnotesize
\caption{System parameter settings}
\label{system}
\begin{tabular}{c|c|c|c}
\hline
\textbf{Parameter} & \textbf{Value} & \textbf{Parameter} & \textbf{Value}\\ 
\hline
${B}$   & 512      & {$N$}    & 8\\
$T$   & 10\,ms   & $\tau$ & 0.125\,ms\\
${BW}$   & 300\,MHz & $\rho$    & 2\\
$p_i$ & [23, 26]\,dBm & $p_0$ & 46\,dBm\\ 
$G_i$ & [1,2]\,dBi      & $G_0$  & 18\,dBi\\   
$N_0$ & -174\,dBm/Hz & $U_i$ & [1,2]\,GB\\
$F_i$ & [100, 200]\,GFLOPS & $F_0$ & 14\,TFLOPS\\ 
$\beta_i$ & [10, 20]\,GB/s & $\beta_0$ & 900\,GB/s\\ 
\hline
\end{tabular}
\end{table}

\subsection{Performance Evaluation}  \label{subsec:performance_eval}
\subsubsection{Performance of Different Training Schemes}

Table \ref{training_time} presents the training time per-round of various models, obtained under the same configurations of the system parameter. Regarding the benchmarks, the comparative analysis includes not only the USL-based two-layer splitting paradigm (including USL, UPSL, USFL and HFSL), but also the vanilla SL-based single-layer splitting paradigm (including PSL, SFL, and EPSL). It should be noted that we also use AC$^2$P$^2$SL based on the single-layer splitting paradigm as a benchmark to evaluate the trade-off between privacy and speed. Furthermore, we consider centralized learning (CL) as an additional baseline, which does not involve model partitioning. In CL, UEs upload raw input data to the BS, where the complete model training process is executed. 

\begin{table}[t]
\small
\centering
\footnotesize
\renewcommand{\arraystretch}{1.5}
\caption{Average per-round training time (s)}
\begin{tabular}{ C{0.35cm} | C{1.4cm} | C{1.2cm} | C{1.2cm} | C{1.35cm} | C{0.8cm} }
\hline
\multicolumn{2}{c|}{Schemes} & ResNet18 & ResNet50 & ResNet101 & ViT \\ \hline
\multirow{4}{*}{SL} 
 & PSL\cite{PSL}  & 3.963 & 4.466 & 5.283 & 6.467 \\ \cline{2-6}
 & SFL\cite{ASFL}  & 3.713 & 4.302 & 5.120 & 6.315 \\ \cline{2-6}
 & EPSL\cite{EPSL} & 3.562 & 3.791 & 4.132 & 4.171 \\ \cline{2-6}
 & {AC$^2$P$^2$SL\footnotemark} & {1.458} & {1.833} & {1.946} & {2.500} \\ \hline
 & USL\cite{SL}   & 5.955 & 6.483 & 7.300 & 10.599 \\ \cline{2-6}
 & UPSL\cite{UPSL}  & 3.971 & 4.450 & 5.316 & 8.877 \\ \cline{2-6}
 & USFL\cite{VUSFL}  & 3.723 & 4.324 & 5.141 & 8.730 \\ \cline{2-6}
 & {HFSL\cite{GAN}}  & {3.238} & {3.870} & {4.687} & {7.851} \\ \cline{2-6}
\multirow{-5}{*}{USL}
 & \textbf{AC$^2$P$^2$SL} & \textbf{2.139} & \textbf{3.024} & \textbf{3.132} & \textbf{3.497} \\ \hline
\multicolumn{2}{c|}{CL} & 0.580 & 1.083 & 1.900 & 3.883 \\ \hline
\end{tabular}
\label{training_time}
\end{table}
\footnotetext{{ AC$^2$P$^2$SL based on single-layer splitting paradigm.}}

\begin{figure}[t]
    \centering
    \includegraphics[width=0.45\textwidth]{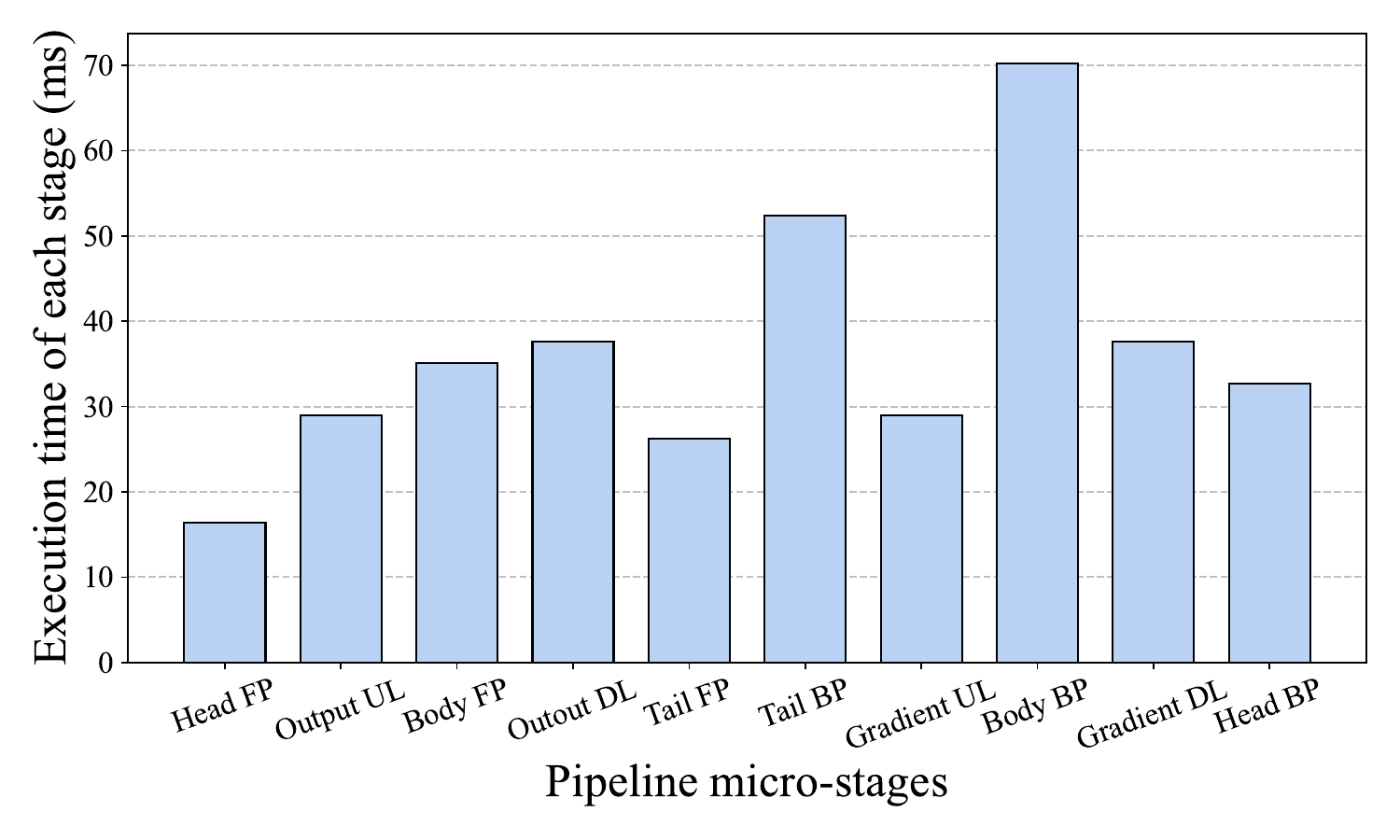}
    \caption{Execution time of pipeline micro-stages for a single micro-batch.}
    \label{times}
\end{figure}

Obviously, whether operating under the single-layer SL or the two-layer USL paradigm, our proposed AC$^2$P$^2$SL framework significantly reduces the training latency across various models compared to their respective baseline methods. This comparison explicitly highlights that the training acceleration enabled by AC$^2$P$^2$SL possesses exceptional effectiveness and generalizability regardless of SL or USL. 
Compared to the USL paradigm, both SL and CL schemes compromise data privacy to varying degrees in exchange for improved training efficiency. However, our proposed AC$^2$P$^2$SL, anchored in USL, significantly reduces training latency even compared to SL methods while fully preserving data privacy.
Notably, this reduction becomes more pronounced as model complexity increases. For instance, when training the ViT model, our scheme achieves a lower training latency than even the CL paradigm. This surge in CL training time can be attributed to the sequential nature of data uploading and subsequent computation processes. For complex models, the substantial computational overhead significantly inflates the total training duration.     
Moreover, a horizontal comparison between AC$^2$P$^2$SL under SL and USL reveals that the single-layer splitting indeed achieve a further reduction in training time by trading off the label privacy. Both of them are significantly better than the baselines at the same privacy level. The single-layer splitting can be taken as a special case of the two-layer splitting when the split layers coincide (i.e. $l_1=l_2$). Therefore, depending on the specific privacy requirements of different deployment scenarios, we can rationally select either the single-layer or two-layer splitting paradigm of AC$^2$P$^2$SL for their training. 

To elucidate the underlying mechanism, Fig. \ref{times} illustrates the duration of the individual micro-stage training ViT introduced in Section \ref{section4.1}. It is evident that, following the SPA optimization, the durations of these micro-stages become relatively balanced. Attributed to the parallelism between communication and computation in the AC$^2$P$^2$SL scheme, the additional computational overhead effectively overlaps with the communication time. Consequently, the proposed AC$^2$P$^2$SL achieves performance that is, perhaps counter-intuitively, superior to CL, thereby demonstrating the significant efficacy of our training approach.

\begin{figure}[t]
\centering
\begin{subfigure}{0.65\linewidth}
  \centering
  \includegraphics[width=\linewidth]{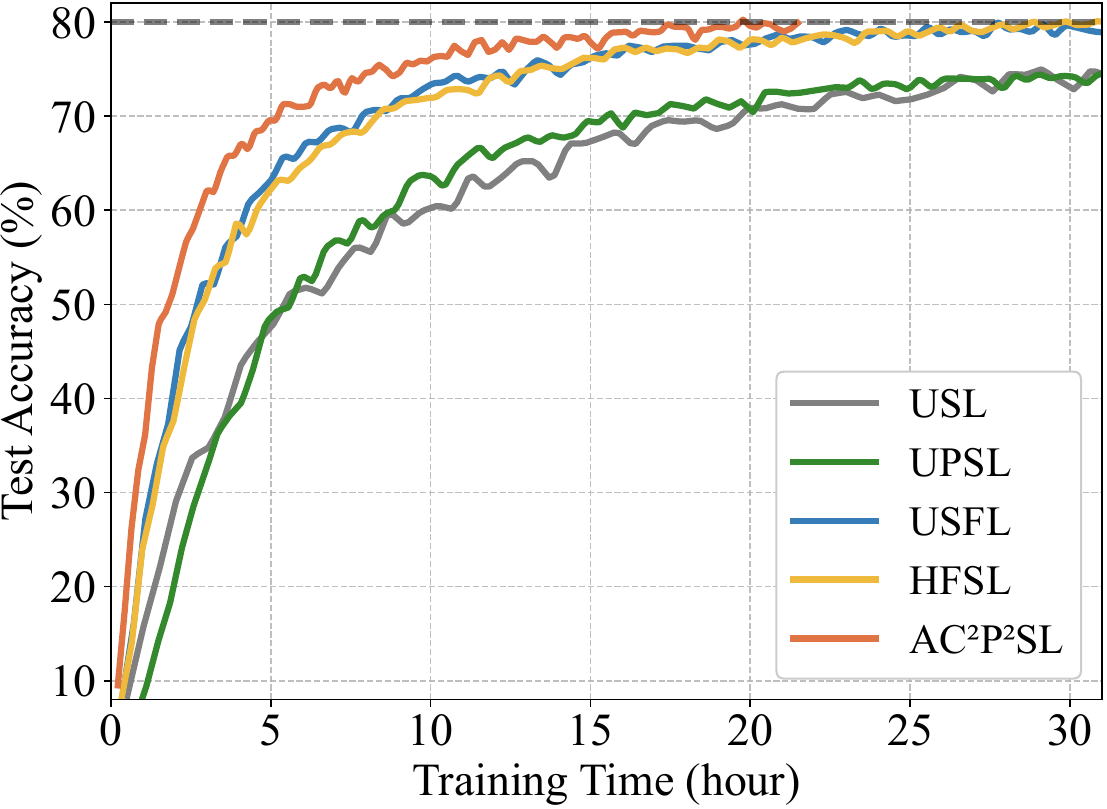}
  \caption{{IID}}
  \label{IID}
\end{subfigure}
\begin{subfigure}{0.65\linewidth}
  \centering
  \includegraphics[width=\linewidth]{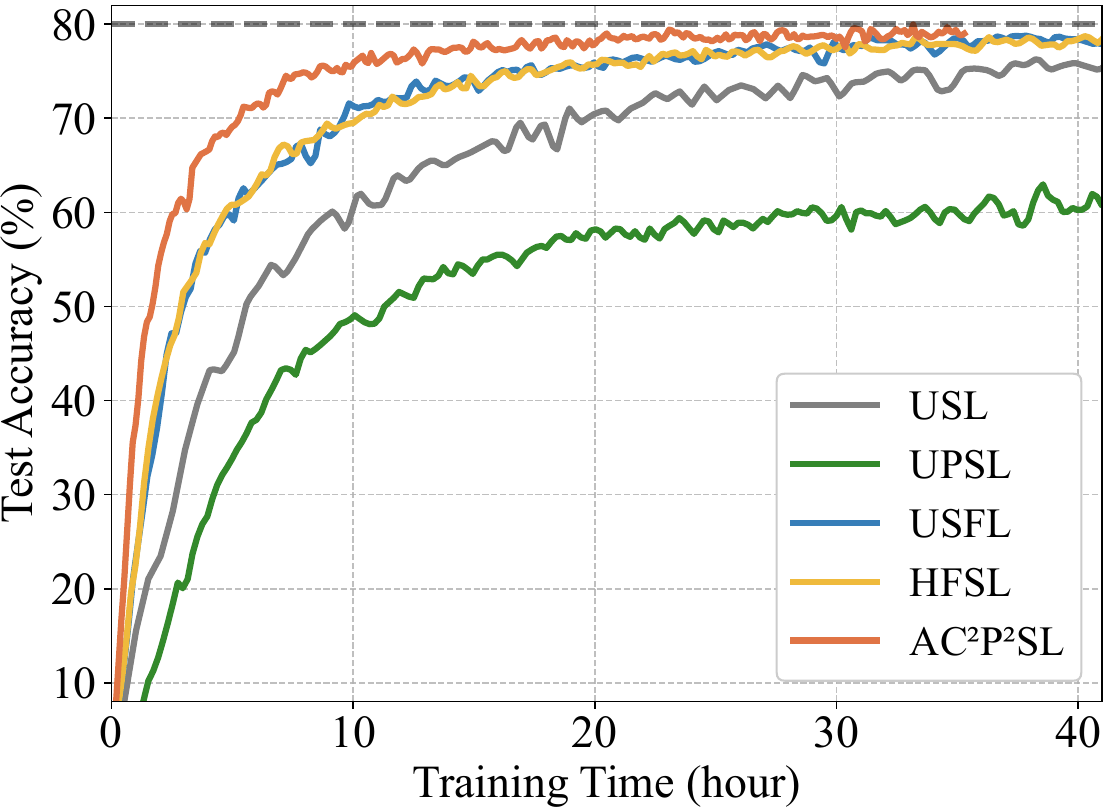}
  \caption{{Non-IID}}
  \label{Non-IID}
\end{subfigure}
\caption{{Test accuracy of ResNet-101 under IID and non-IID settings versus training time.}}
\vspace{-0.3cm}
\label{accuracy}
\end{figure}

Figure \ref{accuracy} presents the test accuracy under IID and non-IID versus training time, training on the ImageNet-100 dataset using ResNet-101 as a representative example to intuitively illustrate the convergence behaviors of USL-based schemes. Under non-IID heterogeneous data distributions, the accuracy is slightly reduced at identical training time stamps compared to the IID setting. However, as the epochs progress, the accuracies gradually converge. In the baseline schemes, UPSL eliminates the parameter sharing of head and tail models among UEs severely restricting the training efficacy, whereas HFSL only performs a parameter update after continuous backpropagation over two consecutive rounds of data. These steps negatively impact their convergence trajectories, ultimately prolonging the total training time required to reach the target accuracy. In contrast, AC$^2$P$^2$SL strictly preserves the mathematical equivalence of the per-round parameter updates, ensuring zero adverse effects on the model convergence process. Consequently, under AC$^2$P$^2$SL, the number of epochs required to converge remains virtually identical to that of standard USFL. The reduction ratio achieved in per-round training time is close to the reduction ratio in the overall training convergence time, as explicitly demonstrated by the accuracy convergence curves.

\begin{figure}[t]
\centering
\begin{subfigure}{0.8\linewidth}
  \centering
  \includegraphics[width=\linewidth]{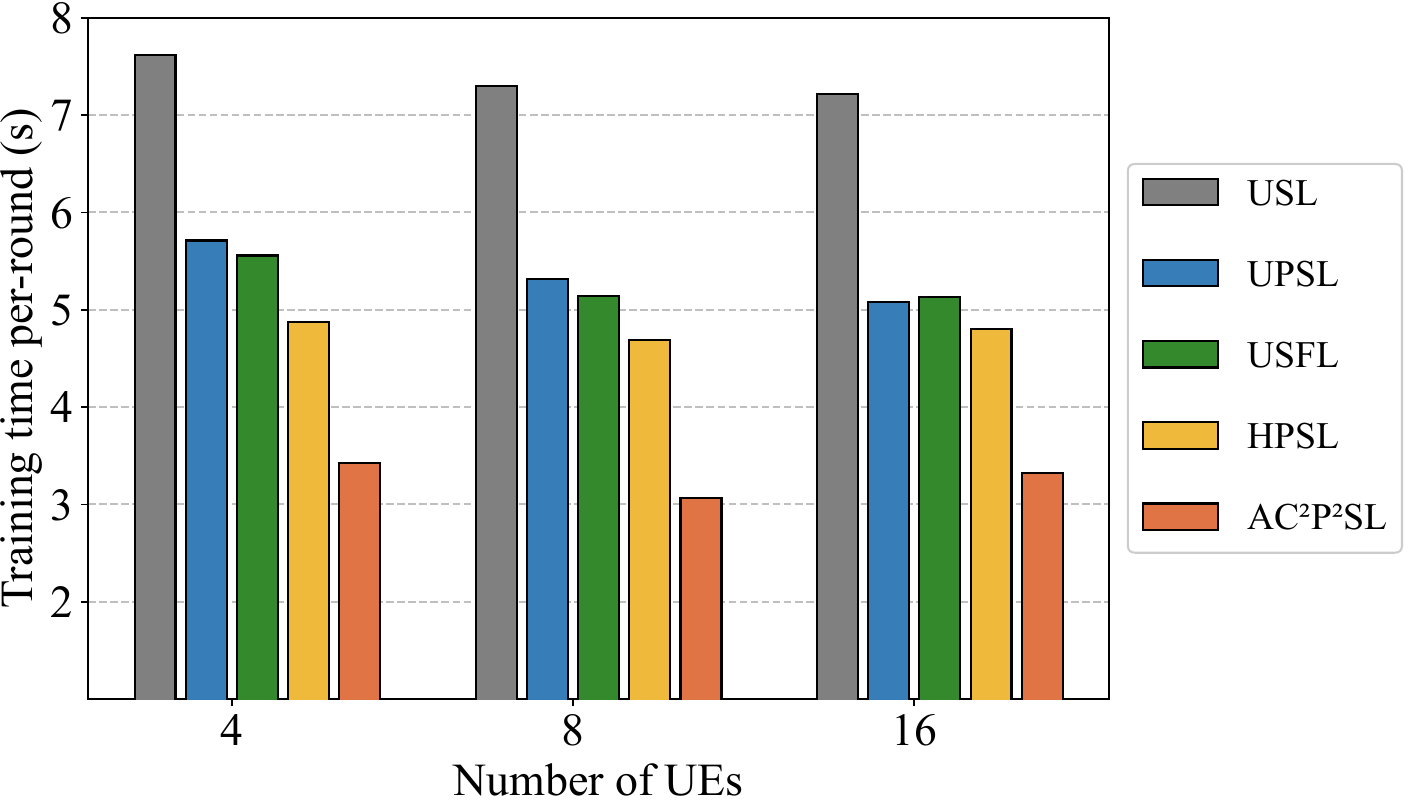}
  \caption{{ResNet-101}}
  \label{num_r101}
\end{subfigure}
\begin{subfigure}{0.8\linewidth}
  \centering
  \includegraphics[width=\linewidth]{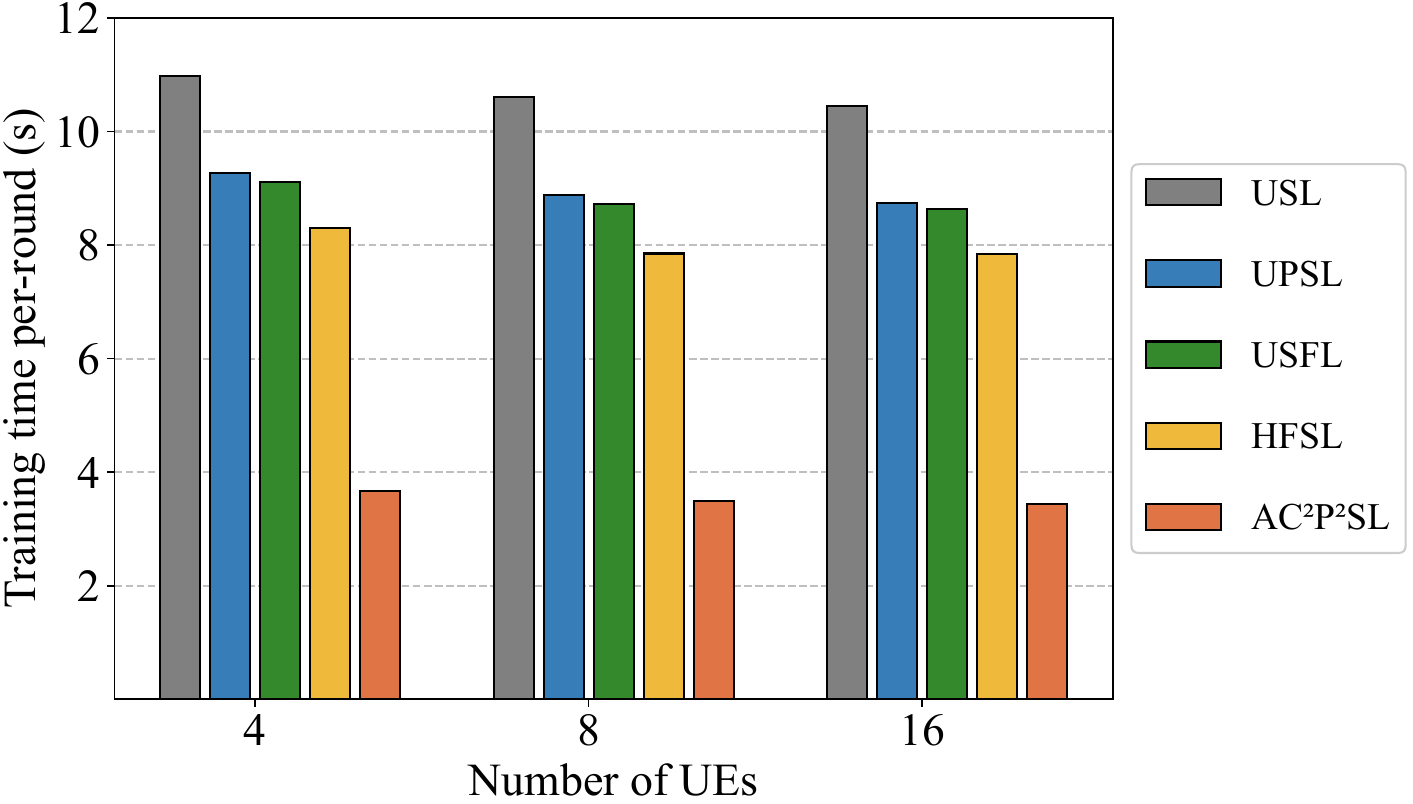}
  \caption{Vision Transformer}
  \label{num_vit}
\end{subfigure}
\caption{Training time under different numbers of UEs.}
\label{number}
\end{figure}

\subsubsection{Performance under Different Numbers of UEs}
{ To demonstrate the scalability of our proposed scheme, we conduct experiments under different numbers of UE. Fig. \ref{number} presents a comparative analysis of the model training latency for various schemes under different UE numbers training ResNet-101 and ViT models. The results unequivocally indicate that AC$^2$P$^2$SL consistently maintains significantly lower training latency compared to the benchmarks, even achieving an average reduction of over 60\% relative to baselines of ViT. This substantial reduction underscores a marked improvement in training efficiency, and the observed consistency validates the robustness of the performance gains yielded by AC$^2$P$^2$SL across different scales of UE deployment.

Given that the global batch size is fixed, an increase in the number of UEs leads to a proportional decrease in the local data batch size allocated to each UE, thereby significantly reducing the local computation time. Conversely, the communication capability of each UE diminishes due to the reduction in the number of allocated time slots per user. The interplay of these opposing factors results in the total training time exhibiting only a marginal decrease as the number of UEs increases.
However, when considering heterogeneous hardware  across different UEs, this general downward trend may experience minor localized fluctuations. As observed in Figure \ref{num_r101} with the number of UEs from 8 to 16, a few UEs with extreme hardware performance limitations might slightly prolong the total time. Nevertheless, the experimental comparisons verify that such fluctuations are overall negligible. This robust stability primarily stems from the rational resource allocation strategy embedded in our SPA algorithm, demonstrating that our framework is highly scalable and resilient to straggler effects.
}

\begin{figure}[t]
\centering
\begin{subfigure}{0.65\linewidth}
  \centering
  \includegraphics[width=\linewidth]{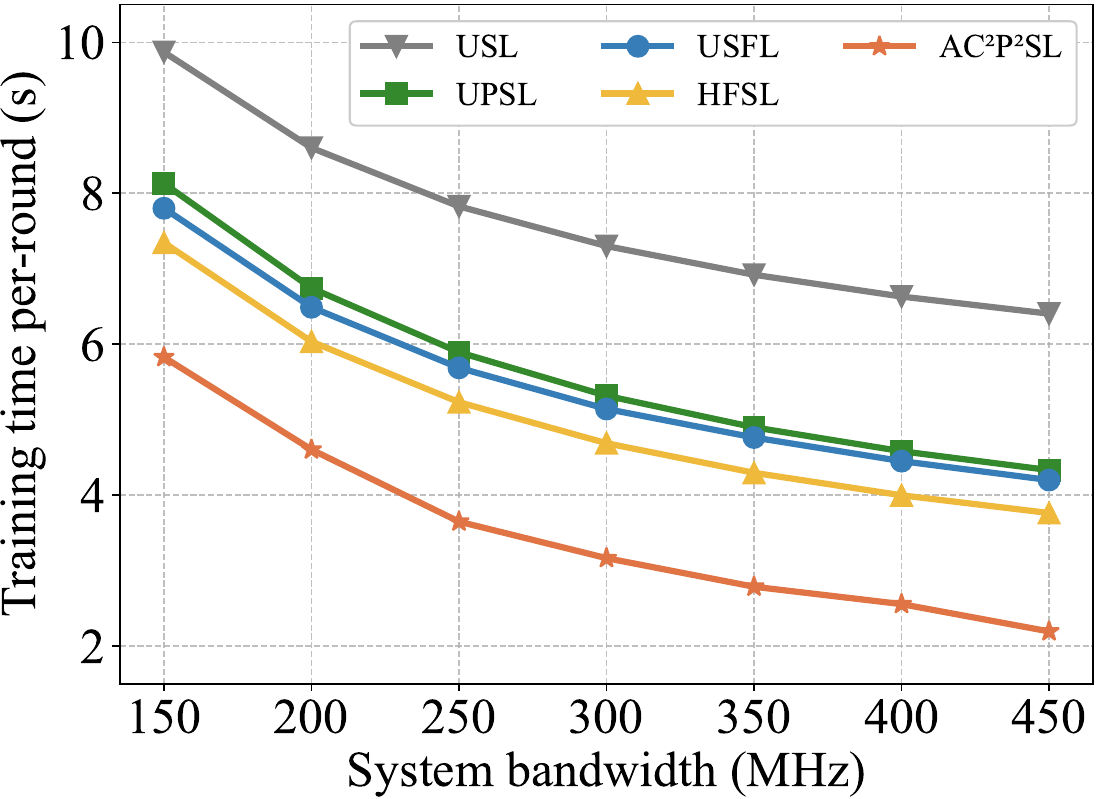}
  \caption{{ResNet-101}}
  \label{band_r101}
\end{subfigure}
\begin{subfigure}{0.65\linewidth}
  \centering
  \includegraphics[width=\linewidth]{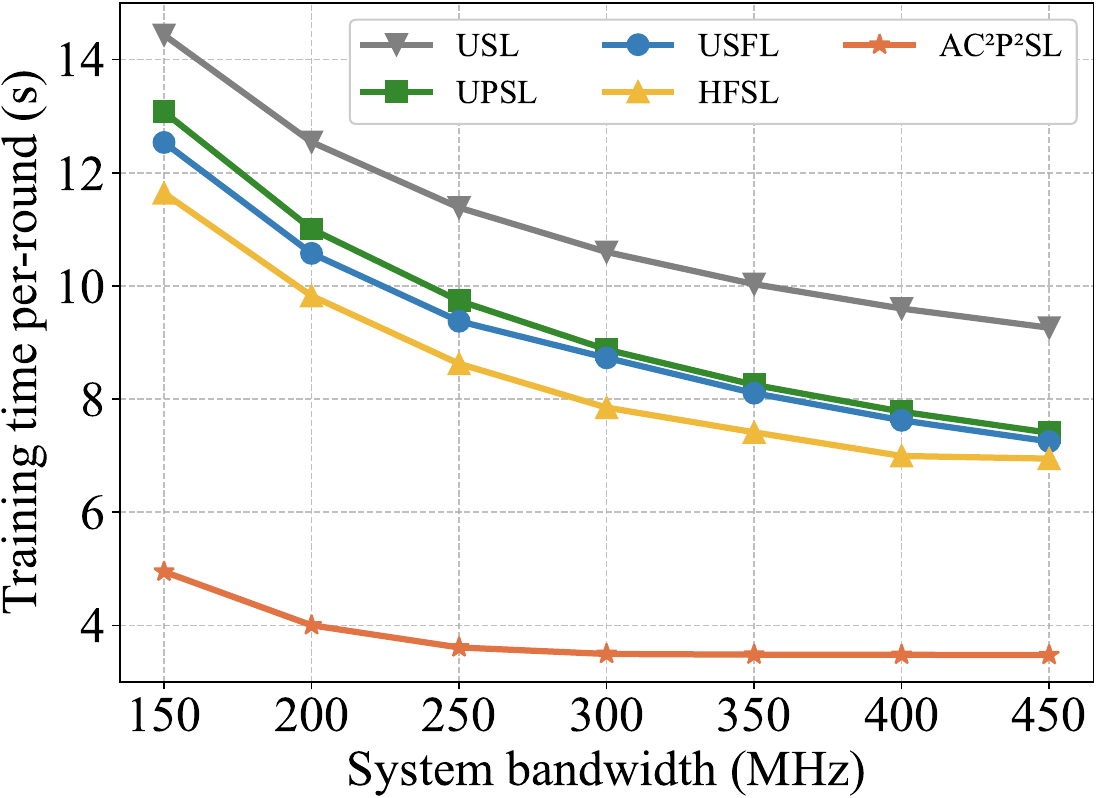}
  \caption{Vision Tranformer}
  \label{band_vit}
\end{subfigure}
\caption{Training time versus the system bandwidth.}
\vspace{-0.2cm}
\label{bandwidth}
\end{figure}

\subsubsection{Robustness against System Bandwidth}
To further validate the robustness of AC$^2$P$^2$SL, we emulate dynamic wireless channel conditions by varying the total system bandwidth, thereby yielding diverse uplink and downlink transmission rates. { Fig. \ref{bandwidth} presents a comparative analysis of the single-round training latency for various USL-based schemes across a range of system bandwidths using ResNet-101 and ViT models. The results indicate that our proposed scheme consistently maintains a huge reduction in training latency across the bandwidth spectrum of 150 MHz to 450 MHz, even all over 50\% in ViT, demonstrating exceptional robustness under varying communication conditions.} Notably, the performance gain is more pronounced in lower system bandwidth regimes when training ViT. This is primarily attributed to the fact that reduced bandwidth results in lower transmission rates and consequently higher communication latency, which significantly prolongs the training duration for sequentially executed benchmark methods. In contrast, leveraging communication-computation pipeline parallelism, the proposed scheme effectively mitigates this latency increase by overlapping tasks. This evidence further substantiates the robustness of our pipeline parallel strategy. In addition, as the bandwidth is large enough, training time gradually converges, indicating that training time is dominated by computation time. In other words, the parallel communication time through the pipeline is masked within the computation time, which significantly improves the training efficiency.

\begin{figure}[t]
\centering
\begin{subfigure}{0.65\linewidth}
  \centering
  \includegraphics[width=\linewidth]{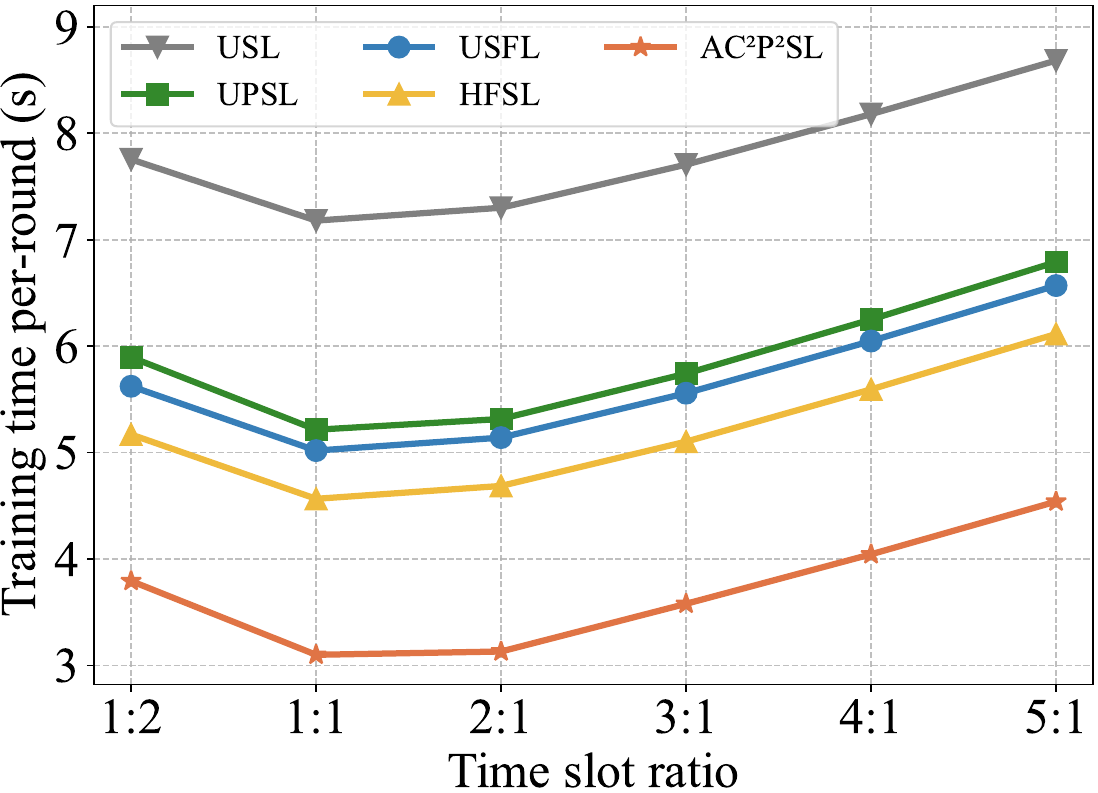}
  \caption{{ResNet-101}}
  \label{slot_r101}
\end{subfigure}
\begin{subfigure}{0.65\linewidth}
  \centering
  \includegraphics[width=\linewidth]{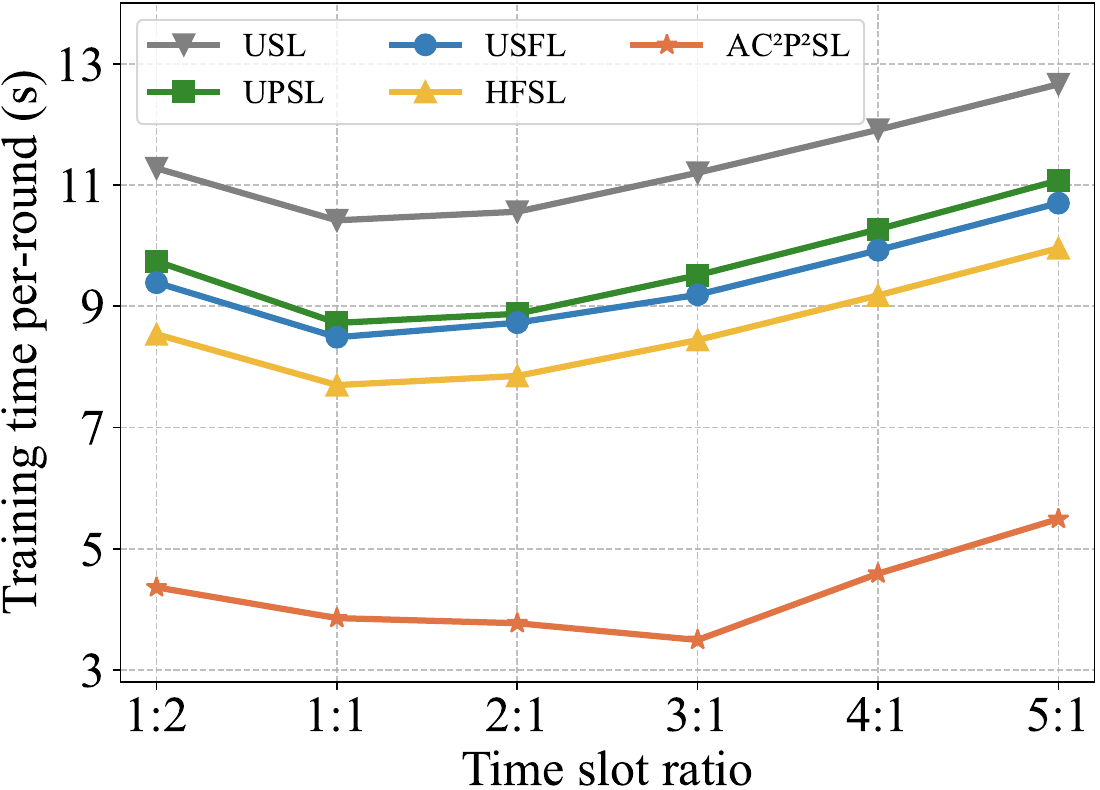}
  \caption{Vision Tranformer}
  \label{slot_vit}
\end{subfigure}
\caption{Training time versus the time slot ratio of uplink to downlink.}
\vspace{-0.2cm}
\label{timeslot}
\end{figure}

\subsubsection{Optimal uplink-downlink time slot ratio}
In our system modeling, we adopted a unified ratio $\rho$, to represent the proportion of uplink to downlink time slots. However, in practical deployments, the communication process is inherently constrained by specific time slot configuration protocols. Consequently, it is necessary to determine the optimal $\rho$ by aligning our model with the configuration rules of practical systems. { Fig. \ref{timeslot} illustrates the per-round training time for various USL-based schemes under different time slot ratios using ResNet-101 and ViT models. In contrast to the conventional USL schemes, which achieve optimal performance at 1:1 ratio, the proposed AC$^2$P$^2$SL framework approaches its optimum at a ratio of approximately 2:1 and 3:1 when training ResNet-101 and ViT, respectively.} This distinct behavior is attributed to the underlying mechanism of pipeline parallelism. In standard pipeline theory, the degree of task overlap is maximized when the durations of respective micro-stages are identical. Consequently, AC$^2$P$^2$SL necessitates adjusting the time slot ratio to equalize the durations of uplink and downlink transmission. Specifically, the system approaches optimal performance when these two transmission times are roughly equivalent. This implies that the configured time slot ratio should be inversely proportional to the theoretical uplink-to-downlink data rate ratio—a theoretical deduction that is entirely consistent with our experimental observations. Furthermore, Fig. \ref{timeslot} also demonstrates that even under extreme ratio configurations, AC$^2$P$^2$SL continues to significantly outperform the comparative benchmarks, highlighting the broad applicability and effectiveness of our proposed scheme.

\subsection{Ablation Studies} \label{subsec:ablation_study}
\begin{figure}[t]
    \centering
    \includegraphics[width=0.8\linewidth]{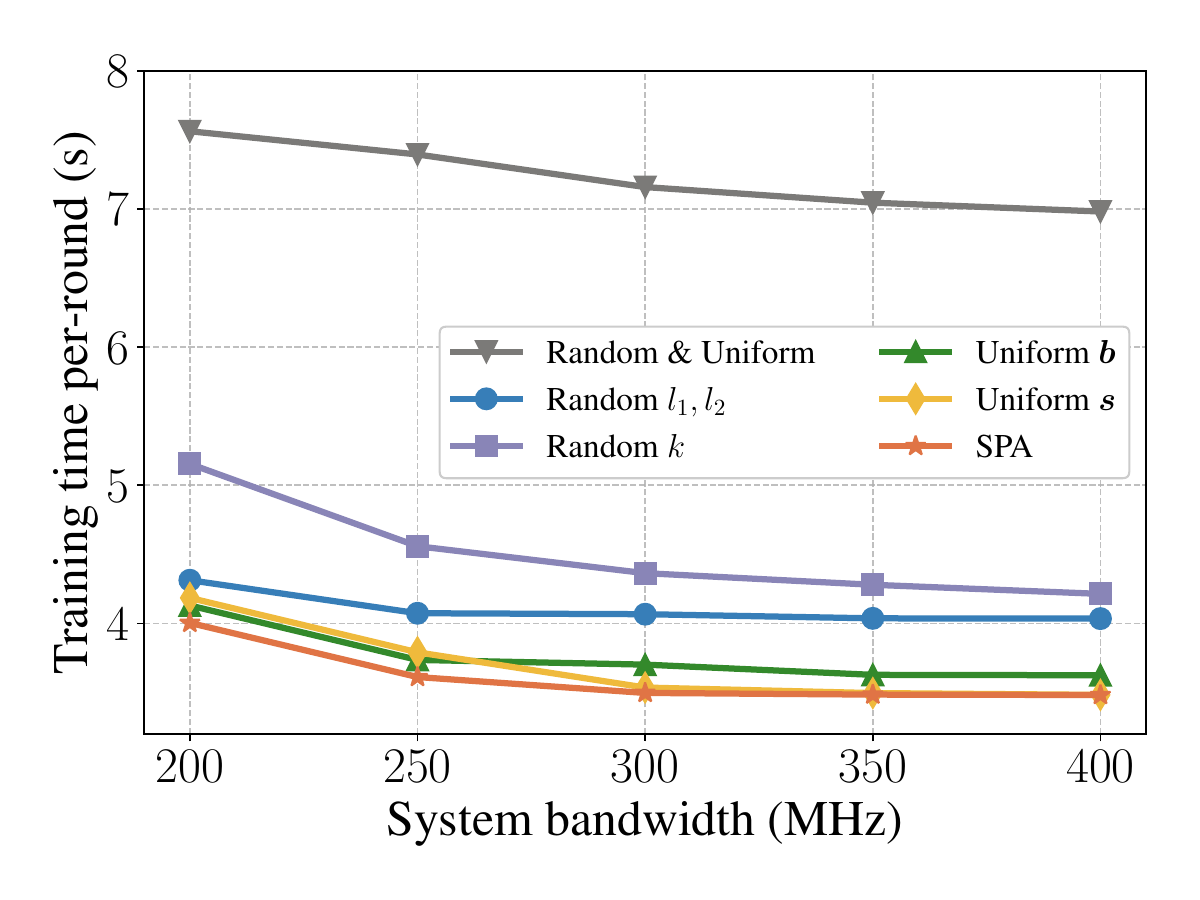}
    \caption{Training time of different baselines under varying bandwidths.}
    \vspace{-0.1cm}
    \label{spa}
\end{figure}
\subsubsection{Ablation Study on SPA Optimization}
To validate the effectiveness of the proposed SPA algorithm during the training process, we conduct an ablation study by comparing it against the following baselines: (i) The split layers $l_1, l_2$ and the number of micro-batches $k$ are determined randomly, while the batch size set $\boldsymbol{b}$ and time slot set $\boldsymbol{s}$ are uniformly allocated across all users; (ii) $l_1, l_2$ are determined randomly, while $k, \boldsymbol{b}, \boldsymbol{s}$ are optimized using the ARA algorithm; (iii) $k$ is determined randomly, while $l_1, l_2, \boldsymbol{b}, \boldsymbol{s}$ are optimized via the SPA; (iv) The batch sizes $\boldsymbol{b}$ are uniformly allocated, while $l_1, l_2, k, \boldsymbol{s}$ are optimized via the SPA; (v) The time slots $\boldsymbol{s}$ are uniformly allocated, while $l_1, l_2, k, \boldsymbol{b}$ are optimized via the SPA. It is noteworthy that, although the split layers and the number of micro-batches in these comparison schemes are determined randomly, they strictly adhere to the constraints defined in (\ref{P1}). Specifically, the selection is confined to the subset of layers that satisfy the pruning strategy. In essence, by systematically isolating specific variables, we investigate their individual impacts on training time.

Fig. \ref{spa} illustrates the single-round training latency for each scheme across varying system bandwidths. While the fully random/uniform strategy yields performance significantly inferior to the other optimized schemes, a cross-comparison with Fig. \ref{bandwidth} reveals a critical insight: even without the SPA, the standalone implementation of the communication-computation pipeline parallel mechanism still maintains a performance advantage over other conventional USL-based paradigms. Observing the other comparison schemes, it becomes evident that the model split points and the number of micro-batches directly dictate the degree of pipeline overlap. Consequently, they exert a more substantial influence on the training duration, a finding that aligns with our theoretical expectations. Furthermore, the SPA demonstrates a distinct performance improvement over methods employing uniform allocation for $\boldsymbol{b}$ and $\boldsymbol{s}$. This improvement is primarily attributed to the handling of device heterogeneity among UEs. In uniform allocation scenarios, faster UEs may induce synchronization latency at the BS side. The SPA effectively mitigates this issue through its alternating optimization process.

\subsubsection{Ablation Study on ARA Optimization}
To further validate the effectiveness of the ARA optimization algorithm, we conduct a comparative analysis of the multi-round training time of the AC$^2$P$^2$SL framework under varying re-allocation thresholds $\delta$. The results are depicted in Fig. \ref{ra}. To ensure the observability of the experimental results, we introduce systematic performance degradation for a specific UE during the 1st, 6th, 11th, and 16th training rounds. Specifically, we incrementally increase the distance between the UE and the BS, while simultaneously decreasing the peak FLOPS, bandwidth, and transmit power. This setup is designed to emulate the realistic fluctuations of communication capabilities induced by UE mobility and the weakening of computational capabilities resulting from local background computing tasks.

In the comparative analysis, the case where $\delta = \infty$ represents the baseline without ARA strategy, whereas $\delta = 0$ implies that re-allocation is executed in every round subsequent to the initial one. For the remaining schemes, the decision to trigger re-allocation is contingent upon whether the temporal variation in the preceding round exceeds the specified threshold. It is evident from the figure that, compared to the baseline $\delta = \infty$, the deployment of our ARA strategy significantly mitigates the adverse impacts of UE performance fluctuations. Furthermore, a lower threshold facilitates more timely resource adjustments, thereby effectively reducing the total training latency. This demonstrates the efficacy of our proposed adaptive resilience mechanism. However, it is worth noting that a lower threshold does not necessarily yield superior performance, as the re-allocation optimization process itself incurs computational overhead. This finding provides empirical guidance for selecting the most appropriate re-allocation threshold to balance optimization gains against computational costs.

\begin{figure}[t]
    \centering
    \includegraphics[width=0.8\linewidth]{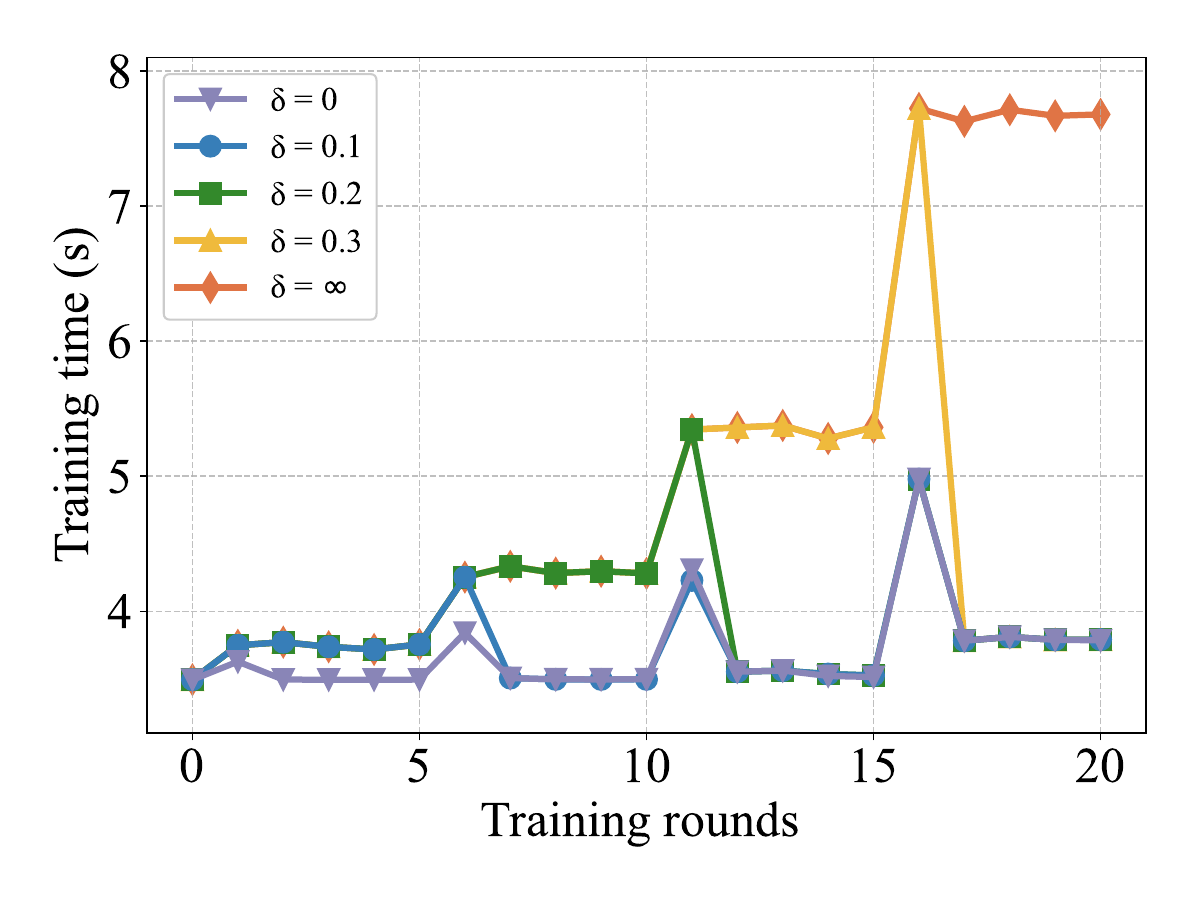}
    \caption{Training time under different re-allocation thresholds with UE parameters randomly changing in specific rounds.}
    \vspace{-0.1cm}
    \label{ra}
\end{figure}

\section{Conclusions}\label{section6}
In this paper, we propose an efficient framework named AC$^2$P$^2$SL for split learning in wireless edge networks. 
Analyzing system constraints across communication, computation, and storage dimensions, we formulate a pipeline training time minimization problem. 
Therefore, we design the SPA algorithm and the ARA strategy to solve this. 
In training workflow, AC$^2$P$^2$SL overlaps the communication and computation tasks within the U-shaped paradigm by communication-computation pipeline parallelism, significantly reducing training time. 
Simulation experiments validate the effectiveness and robustness under different system parameters, and ablation experiments also demonstrate the application effects of pre- and re-allocation algorithms.

In our work, by introducing pipeline parallelism into wireless transmission, communication can be initiated in advance before computation is complete. This overlap reduces waiting delay, which is applicable to any different multiple access transmission method.
Notably, we observe that for complex models, our approach not only ensures strict privacy preservation but also outperforms centralized learning to a certain extent. This finding offers a compelling and effective paradigm for collaborative model training between edge servers and terminal devices.

\bibliographystyle{IEEEtran}
\bibliography{IEEEabrv.bib, myabrv.bib, ref.bib}

\end{document}